\documentclass[11pt,letterpaper]{article}

\usepackage[margin=1in]{geometry}
\usepackage{theorem,latexsym,graphicx,amssymb}
\usepackage{amsmath,enumerate}
\usepackage{bm,bbm}
\usepackage{wrapfig}
\usepackage{subfigure}
\usepackage{float}
\usepackage{epsfig}
\usepackage{xspace}
\usepackage{paralist}
\usepackage{enumerate}
\usepackage{cases}
\usepackage{caption}
\usepackage{algorithm}
\usepackage[noend]{algpseudocode}
\usepackage{xcolor}
\usepackage{complexity}
\usepackage{multicol}
\usepackage{graphicx}
\usepackage{thm-restate}

\newenvironment{proof}{{\bf Proof:  }}{\hfill\rule{2mm}{2mm}\vspace*{5pt}}

\numberwithin{figure}{section}
\numberwithin{equation}{section}
\newtheorem{definition}{Definition}[section]

\newtheorem{example}{Example}[section]

\newtheorem{theorem}{Theorem}[section]
\newtheorem{lemma}{Lemma}[section]
\newtheorem{claim}{Claim}[section]

\usepackage{multirow}

\newcommand{\bE}{\mathbb{E}}
\newcommand{\bsigma}{\bm{\sigma}}
\newcommand{\by}{\bm{y}}
\newcommand{\etal}{{\em et~al.~}}
\newcommand{\bbE}{{\mathbb{E}}}
\newcommand{\bzero}{{\mathbf{0}}}
\newcommand{\bD}{{\bm{D}}}
\newcommand{\bb}{{\bm{b}}}

\title{Eliciting Truthful Reports with Partial Signals in Repeated Games}

\author{
Yutong Wu \thanks{The University of Texas at Austin. {\texttt{yutong.wu@utexas.edu}}}
\and 
Ali Khodabakhsh \thanks{The University of Texas at Austin. {\texttt{ali.kh@utexas.edu}}}
\and
Bo Li\thanks{The Hong Kong Polytechnic University. {\texttt{comp-bo.li@polyu.edu.hk}}} 
\and 
Evdokia Nikolova \thanks{The University of Texas at Austin. {\texttt{nikolova@austin.utexas.edu}}}
\and 
Emmanouil Pountourakis \thanks{Drexel University. {\texttt{manolis@drexel.edu}}}
}
\date{}

\begin{document}

\begin{titlepage}
	\thispagestyle{empty}
	\maketitle
	
	\begin{abstract}
We consider a repeated game where a player self-reports her usage of a service and 
is charged a payment accordingly by a center.  
The center observes a partial signal, representing part of the player's true consumption, 
which is generated from a publicly known distribution.
The player can report any value that does not contradict the signal and the center issues a payment based on the reported information.
Such problems find application  in  net metering billing in the electricity market, where a customer's actual consumption of the electricity network is masked  and  complete verification is impractical.
When the underlying true value is relatively constant, we propose a penalty mechanism that elicits truthful self-reports. 
Namely, besides charging the player the reported value, the mechanism charges 
a penalty proportional to her inconsistent reports. 
We show how fear of the uncertainty in the
future incentivizes the player to be truthful today.
For Bernoulli distributions, we give the complete analysis and optimal strategies given any penalty. 
Since complete truthfulness is not possible for continuous distributions,
we give approximate truthful results by a reduction from Bernoulli distributions. 
We also extend our mechanism to a multi-player cost sharing setting and give equilibrium results.
	\end{abstract}
\end{titlepage}

\section{Introduction}


Consider the following repeated game where a center owns resources and one or more strategic players pay the center to consume the resources. 
In every round, a player self-reports their usage, which  will then be used to determine their payment to the center.
However, it is not always possible for the center to verify the submitted information from  the players. 
Instead, only part of the actual consumption is revealed to the center based on some publicly known distribution. 
 A player can report any value that is at least the revealed amount.
Without any external interference, a player will naturally report exactly the revealed amount (potentially lower than the true consumption) to minimize their payment. 
The center then needs to determine a payment mechanism such that each player is incentivized to report their true value. 

The electricity market is facing precisely the described problem. 
As the number of electricity {\em prosumers} increases each year, new  rate structures are designed to properly calculate the electricity bill  for this special type of consumers while ensuring that every customer is still paying their fair share of the network costs. 
Prosumers are those who not only consume energy but also produce electricity via distributed energy resources such as rooftop solar panels.
Among different rate structures, {\em net metering} is a popular billing mechanism that is currently adopted in more than 40 states in the US \cite{ncsl}. 
Net metering 
charges prosumers a payment proportional to their net consumption, i.e., gross consumption minus the production \cite{net-metering}, demonstrated in 
Fig.~\ref{fig:net-metering}.
The payment includes the electricity usage as well as grid costs that are incurred by using the electricity network.

\begin{figure}[h]
    \centering
    \includegraphics[width=.6\columnwidth]{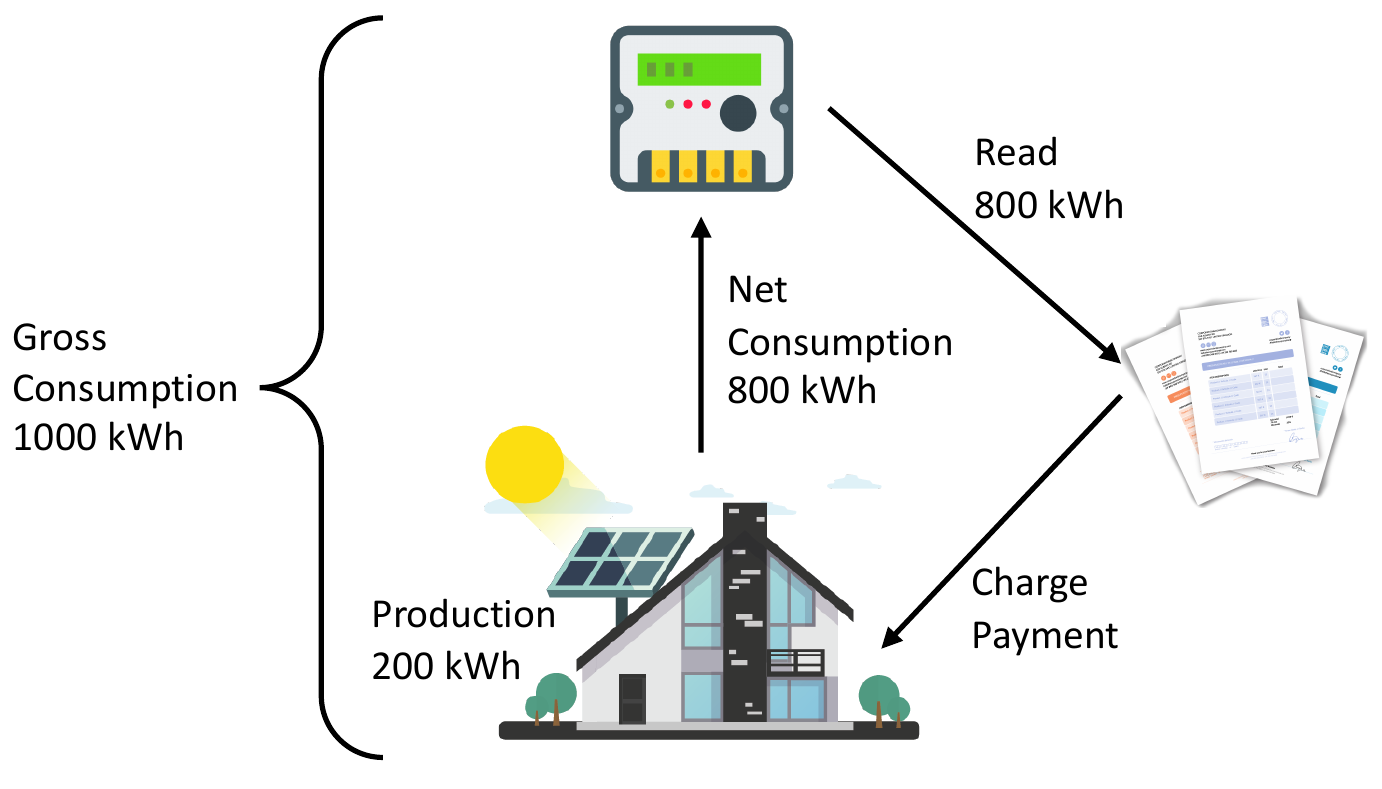}
    \caption{Net metering for electricity prosumers.}
    \label{fig:net-metering}
\end{figure}

The controversy in net metering lies in that 
prosumers fail to pay their  share of the grid costs when they do  not  have local storage equipment   \cite{gautier2018prosumers}.
In the United States, only 4\% of the solar panel owners also own the battery to store the produced solar energy \cite{4percent}. 
For those who do not own battery storage, the generated  power has to be transmitted back to the grid.  Accordingly, the daily consumption of power by these prosumers also needs to come from the grid instead of directly from the solar panels.  
In this way, most prosumers have under-paid their share of the network costs and  become ``free-riders'' of the electricity grid.
The grid is often subject to costly line upgrades and  net metering unevenly shifts such costs to traditional consumers, who usually come from lower-income households \cite{hoarau2019network}. 
Indeed, previous research works have suggested that prosumers should pay a part of the grid costs proportionally to their gross consumption, not net consumption  \cite{gautier2018prosumers,khodabakhsh2019prosumer}.
However,  the gross consumption is  hidden from the utility companies since only net consumption can be observed from the meter.
Meanwhile, there is no incentive for prosumers to voluntarily report their true consumption as it will only increase their electricity bills.




Fortunately, the production from solar panels  usually  follows some pattern while the gross consumption of electricity for a typical household  stays relatively constant,
which is especially true for industrial sites -- the major consumers for utilities \cite{nadel2014electricity,eia}. 
Thus, the observed consumption can be assumed to follow some natural distribution and the center is able to detect dishonesty when a player's report differs from their reporting history. 
With this idea, we propose a simple penalty mechanism, the {\em flux mechanism}, that elicits truthful reports from players in a repeated game setting when only partial verification is possible. 
Particularly, a player is charged their reported value as well as a penalty due to inconsistency in consecutive reports in each round. 
 The main goal is to ensure that every player reports their true values  and no penalty payment is collected. 
 We show that the combining effect of (i) the penalty rate and (ii) the length of the game is sufficient for inducing truthful behavior from the player for the entire game. 
As the horizon of the game increases, the minimum penalty rate for truth-telling as an optimal strategy decreases.
In other words, it is the fear of the uncertainty in the future that incentivizes the player to be truthful today. 
%
%
%
%


\subsection{Our Contribution}
We address the  problem of eliciting truthful reports when the center sees part of the player's private value based on some publicly known distribution.
The strategic player reports some value that is at least the publicly revealed value and is charged a payment accordingly.  
%
%
%
%
We propose a truth-eliciting mechanism, {\em flux mechanism}, 
that utilizes the player's fear of uncertainties to achieve truthfulness. 
In the first round, the player is charged a ``regular payment'' proportional to the consumption they report. 
Starting from the second round, besides 
the regular payment, the player is charged a ``penalty payment'', which is $r$ times the (absolute) difference between the reports in the current and the previous rounds, where the penalty rate $r$ is set by the center before the game starts.  

Intuitively, a player can save their regular payment by under-reporting their consumption, but they will then face the uncertainty of paying penalties in the future rounds due to inconsistent reports. 
Under most settings, if $r$ is set to be infinitely high, the players will be completely truthful 
to avoid any penalty payment. 
However, a severe punishment rule 
is undesirable and discourages players from participating. 
Therefore, we want to understand the following question.
\begin{quotation}
    What is the minimum penalty rate such that the player is willing to report their true value?
\end{quotation}

We observe that no finite penalty can achieve complete truthfulness for arbitrary distributions as a player's true consumption may never be revealed exactly. 
We can, however, obtain approximate truthfulness for a general distribution from analyzing complete truthfulness for a corresponding Bernoulli distribution. 
For $Ber(p)$, the partial signal equals to the true consumption with probability $p$ and 0 with probability $1-p$ for $p \in(0,1)$.
We give results for Bernoulli distributions in Main Results 1 and 2. 
For arbitrary distributions, we redefine $p$ as the probability for having a partial signal that is at least $\alpha$ times the true consumption, for $\alpha\in[0,1]$, to obtain $\alpha$-truthfulness (Main Result 3).




\medskip

\noindent{\bf Main Result 1} (Theorem \ref{thm:ber-p-no-hist})
{\em For a $T$-round game with Bernoulli distribution $Ber(p)$, 
the player is completely truthful if and only if the penalty rate is at least
\[
\frac{1-(1-p)^T}{p-p(1-p)^{T-1}}.
\]
}

Main Result 1 gives the minimum penalty rate that guarantees complete truthfulness for $Ber(p)$ distributions. 
We also want to understand how players would behave if the penalty rate is not as high, which describes the situation when the center is willing to sacrifice some degree of truthfulness by lowering the penalty rate.
Given any penalty rate, we show that 
a player's optimal strategies can be described as one or a combination of three basic strategies, {\em lying-till-end, lying-till-busted} and {\em honest-till-end}.
Specifically, with a low penalty rate, the player is always untruthful to save regular payment, i.e., lying-till-end is optimal. 
As the penalty rate increases, the player's optimal strategy {\em gradually} moves to lying-till-busted, which is to be untruthful
until the partial signal is revealed as the true consumption for the first time 
and then stays truthful
for the rest of the game.
When the penalty rate is sufficiently high, the player would avoid lying completely and 
reports the truth, i.e., she is honest-till-end. 

\begin{table}[htbp]
\renewcommand{\arraystretch}{2}
    \centering
    \begin{tabular}{|c|c|c|}
            \hline
            Bernoulli Prob. & Penalty Rate & Optimal Strategy\\
            \hline
            \multirow{5}{*}{$p\ge 0.5$} & $r\le \frac{1}{2p}$     &  lying-till-end\\\cline{2-3}
            & \multirow{2}{*}{$\tfrac{1}{2p} < r \le 1$}    & lying-till-busted\\
            && + lying last round \\\cline{2-3}
            & $1< r < \frac{1-(1-p)^T}{p-p(1-p)^{T-1}}$ & lying-till-busted\\\cline{2-3}
            & $r \ge \tfrac{1-(1-p)^T}{p-p(1-p)^{T-1}}$ & honest-till-end\\\hline
            \multirow{5}{*}{$p< 0.5$} & $r\le 1$     &  lying-till-end\\\cline{2-3}
            & \multirow{2}{*}{$h(t-1) < r \le h(t)$}    & lying-till-end first $t$ rounds\\
            && + lying-till-busted for rest \\\cline{2-3}
            & $h(T-1) < r < \tfrac{1-(1-p)^T}{p-p(1-p)^{T-1}}$ & lying-till-busted\\\cline{2-3}
            & $r \ge \tfrac{1-(1-p)^T}{p-p(1-p)^{T-1}}$ & honest-till-end\\\hline
            \end{tabular} 
    \caption{Optimal strategy given penalty rate $r$ under $Ber(p)$ distributions}
    \label{tab:ber:opt-str}
\end{table}

\medskip

\noindent{\bf Main Result 2} (Theorems \ref{thm:ber-p-no-hist} and \ref{thm:single-thres-with-hist})
{\em For a $T$-round game with Bernoulli distribution $Ber(p)$, given any penalty rate $r$, the player's optimal strategy is summarized in Table \ref{tab:ber:opt-str}, where
\begin{align*}
    h(t) = \frac{1-(1-p)^t}{2p-p(1-p)^{t-1}}, \text{ for $1 \le t \le T$.}
\end{align*}
}


For arbitrary distributions, including  uniform distributions,  
 it is impossible to obtain complete truthfulness without setting penalty to infinity. 
Main Result 3 gives a reduction from Bernoulli distributions to general distributions for approximate truthfulness. 

\medskip

\noindent{\bf Main Result 3} (Theorem \ref{thm:single-approx})
{\em
Given $\alpha \in [0,1]$ and an arbitrary distribution with CDF $F$, 
if a penalty rate $r$ achieves complete truthfulness for $Ber(p)$ where $p=1-F(\alpha D)$ and $D$ is the player's true gross consumption,  
then the same $r$ achieves $\alpha$-approximate truthfulness for distribution $F$.
}
\medskip

Finally, we extend our results to multiple players.
We note that if the players are charged independently, applying the flux mechanism to each individual elicits truthful reports.
A more complicated but realistic setting, especially in net metering, is the cost sharing problem where the players split an overhead cost based on their submitted reports.
We propose the {\em multi-player flux mechanism} where the penalty payment is the same as before but the regular payment is now a share of some overhead cost. 
Again, if the penalty rate is sufficiently high, the players stay truthful, regardless of others' behavior, to avoid any penalty payment, 
i.e., the truthful report profile forms a {\em dominant strategy equilibrium}.  
As the penalty rate decreases, the truthfulness of a player may  depend on other players' actions. 
That is, with a lower penalty rate, a truthful report profile forms a {\em Nash equilibrium}.
For both equilibrium definitions, we are interested in the following question. 
\begin{quote}
    What is the minimum penalty rate for the truthful report profile to form a dominant strategy or Nash equilibrium?
\end{quote}
We give exact penalty thresholds for both truthful equilibria under Bernoulli distributions and use a reduction to obtain approximate results under arbitrary distributions in Main Result 4.

\medskip

\noindent{\bf Main Result 4} (Theorems \ref{thm:multi-ber-dom}, \ref{thm:multi-exist},  \ref{thm:multi-gen-dom} and \ref{thm:multi-gen-ne})
{\em
For any $T$-round game with distribution $Ber(p)$, truthful strategy profile is a dominant strategy equilibrium if and only if 
\[
r\ge \frac{C}{nD}\frac{1-(1-p)^{n-1}}{p}\frac{1-(1-p)^T}{p-p(1-p)^{T-1}},
\]
and a Nash equilibrium if and only if 
\[
r\ge \frac{C}{nD}\frac{1-(1-p)^T}{p-p(1-p)^{T-1}}.
\]
Given $\alpha\in[0,1]$ and any distribution with cumulative distribution function $F$,
let $p=1-F(\alpha D)$, where $D$ is the true gross consumption.
Then $\alpha$-approximate truthful profile is a Nash equilibrium if 
\[
r\ge \frac{1}{\alpha}\frac{C}{nD}\frac{1-(1-p)^T}{p-p(1-p)^{T-1}},
\]
and the  $\alpha$-approximate truthful profile is a dominant strategy equilibrium if  
\[
r\ge \frac{1}{\alpha}\frac{C}{nD}\frac{1-(1-p)^{n}}{p}\frac{1-(1-p)^T}{p-p(1-p)^{T-1}}.
\]


}

\subsection{Related Works}

Unfairness in  net metering  is a reflection of the famous free-rider problem, where an individual consumes a good but fails to pay or under-pays their share \cite{musgrave1959theory}. 
Possible solution schemes for overcoming the free-rider problems include government taxation \cite{groves1977optimal,pasour1981free}, appealing to altruism \cite{hindriks2002free,laury2008altruism}, and privatization \cite{maskus2004globalization,volokh2010privatization}. 
Researchers in the past have analyzed free-rider problems under the context of 
blood donation \cite{abasolo2014blood}, healthcare reforms \cite{cutler1995cost}, climate change \cite{ostrom2009polycentric}, etc. 
A few papers, particularly, have identified the free-rider problem in net metering for electricity  prosumers \cite{hoarau2019network,khodabakhsh2019prosumer,koumparou2017configuring,negash2015combined}.
Our work follows  \cite{khodabakhsh2019prosumer}, where a primitive version of the penalty mechanism is  proposed for promoting a fairer electricity rate structure. 
We formally define the mechanism and provide the corresponding theoretical analysis. 

Our setup is also similar to the public goods game, where each player is to make a contribution to a public ``pool'' that will then be re-distributed.
Without external measures, contributing zero to the pool is a Nash equilibrium \cite{archetti2011coexistence}.
An application of the public goods game is the famous ``tragedy of the commons'' \cite{hardin1968tragedy}. 
Reward or punishment schemes are introduced to incentivize the players to contribute the full amount \cite{brandt2003punishment,dong2016dynamics,nockur2021different}. 
It is  suggested that truthful players are willing to punish free-riders  in a public goods game  \cite{fehr2000cooperation}. 
 Another solution to the public goods game is to track the reputation of each player \cite{milinski2002reputation}. 

Another related line of works is information elicitation with limited verification ability.
\cite{caragiannis2012mechanism} and \cite{ball2019probabilistic} worked on probabilistic verification where a lying player may be caught by a probability based on her type.   
Strategic classification considers the problem where individuals manipulate their input to obtain a better classification outcome \cite{hardt2016strategic}. 
For strategic classification, a number of mechanisms are  proposed to inventivize truthful behavior or maximize social welfare \cite{haghtalab2020maximizing,liang2020data,zhang2021incentive}.

\section{Problem Statement}
In this section, we formally define our problem under the single player setting and defer the extension to multiple players to Section \ref{sec:multi}.
The player has a gross consumption 
$D \ge 0$, which is her private information. 
The game has $T$ rounds where $T > 1$ as otherwise the flux mechanism becomes invalid. 
In each round $t$, the center observes a partial signal,  $y_t \le D$,  which is randomly and independently drawn from a distribution $F$ supported on $[0,D]$. 
We use $r \ge 0$ to denote the penalty rate.
In a flux mechanism, 
a player cares more about the number of rounds left in the future rather than the number of rounds has passed. 
Thus we use $t = T, T-1, \cdots, 1$ to denote the current round, where {\em $t$ means there are $t$ rounds left}, including the current round.
For example, the first round is round $T$, the last round is  round $1$, and the previous round 
of round $t$ is round $t+1$.
For round $t\le T$, the flux mechanism runs as follows.
\begin{itemize}
    \item The center observes the player' net consumption $y_t \sim F$.
    \item The player submits their  reported gross consumption which is at least the net consumption, $b_t \ge y_t$.
    The player may not be truthful, i.e., $b_t$ may not equal to $D$.
    \item When $t< T$, the player's payment consists of regular payment $b_t$ and penalty payment  $r \cdot |b_{t+1}-b_{t}|$.
    When $t=T$, the player only pays the regular payment.
        

        

\end{itemize}

For $t<T$, we call $b_{t+1}$ the {\em history} of round $t$.\footnote{The {\em history} usually refers to the record from the beginning of the game till the current round. In our mechanism, the history before yesterday does not affect the player's action for today. Therefore, the history in round $t$ only needs to be the report for the previous day.}   
In each round $t$, the player wants to pay the lowest expected total 
payment by reporting $b_t$ without knowing the partial signals for future rounds. 
We call a mechanism {\em truthful} if the player reports $D$ for all rounds.
When two reports bring the same expected payment, we break tie in favor of truthfulness.
We adopt the assumption from Khodabakhsh~\etal  \cite{khodabakhsh2019prosumer} that $D$ does not vary with $t$. We explain in Appendix~\ref{app:const} 
an easy extension where $D_t$ 
is drawn from a known range $[\underline{D},\overline{D}]$.








\section{Bernoulli Distributions}\label{sec:single}
We start with the analysis of Bernoulli distribution as we show later a reduction from an arbitrary distribution to a  Bernoulli distribution.  
We prove it is only optimal for a player to report zero or their true consumption in each round.
The optimal strategies can then be characterized by three basic strategies (Definition~\ref{def:bs}).  
The penalty thresholds are computed by comparing
the different combinations of the basic strategies. 
Due to space limit, we defer most proofs to Appendix~\ref{app:pfs} and focus on explaining the intuition in this section.

\subsection{Basic Strategies}\label{sec:single-ber}


In a Bernoulli distribution setting, 
in each round $t$, the partial signal $y_t$ is $D$ with probability $p$ and $0$ with probability $1-p$.
When the  partial signal equals to the private value, i.e., $y_t = D$, we say that the player is {\em ``busted"} in round $t$.
We first define three basic strategies. 

\begin{definition}[Basic Strategies]\label{def:bs}
For Bernoulli distributed net consumption $y_t \sim Ber(p)$, we define the following as the three basic strategies:
\begin{itemize}
    \item lying-till-end: Report $b_t = 0$ when $y_t = 0$ and $b_t = D$ otherwise;

\item lying-till-busted: Report $b_t=0$ until $y_t = D$ for the first time, then report $D$ for all future rounds;

\item honest-till-end: Report $b_t = D$ for all rounds.
\end{itemize}


\end{definition}



We note that a player's optimal strategy for a given penalty rate $r$ can be solved by
backward induction. 
Let $OptCost(t,r,b_{t+1})$ denote the optimal expected cost for a player starting in round $t$ with penalty rate $r$ 
and report $b_{t+1}$ for the previous round. Then
\begin{align*}
    OptCost(t, r,b_{t+1}) = \min_{b_t} ExpCost(t,r, b_{t+1}, b_t),
\end{align*}
where $ExpCost(t,r,b_{t+1},b_{t})$ is the expected cost for the player starting in round $t$ and reporting $b_t$ (if she is allowed to), with penalty rate $r$ and history $b_{t+1}$, i.e.,
\begin{align*} 
    &ExpCost(t, r, b_{t+1}, b_t) \\
    &=\bE_{y_t}[\max\{y_t, b_t\} + r|\max\{y_t, b_t\} - b_{t+1}|+ OptCost(t-1, r, \max\{y_t, b_t\})] \\
    &= p\big( D + r(D-b_{t+1}) + OptCost(t-1, r,D)\big) + (1-p)\big(b_t + r|b_t - b_{t+1}| + OptCost(t-1, r,b_t)\big).
\end{align*}
The first term on the right-side of the equation above refers to the cost when the partial signal is revealed as $D$ and the player has to report $D$.
The second term refers to the cost when the partial signal is 0 and the player chooses to report $b_t$.
Let $OptCost(0,r,b_1) = 0$ for all $b_1$. 
When $t=T$, i.e., the first round, there is no history $b_{T+1}$.
Therefore, the player simply wants to minimize the following total cost, 
\begin{align*}
    OptCost(T,r)&=\min_{b_T} ExpCost(T,r,b_T)\\
    & =p(D+OptCost(t-1,r,D))+(1-p)(b_T+OptCost(t-1,r,b_T)).
\end{align*}
Solving the recursion will give the characterization of optimal strategies in Table \ref{tab:ber:opt-str}, as we demonstrate in Appendix~\ref{app:recur}. 
In what follows, we discuss a surprisingly simpler and more constructive proof by exploiting the properties of the flux mechanism, which may be of independent interest.




\subsection{Main Theorems}


We observe that there are two key elements that influence the decision making of the player.
\begin{itemize}
    \item [(1)] The player's history, $b_{t+1}$ for $t < T$. 
    The value of $b_{t+1}$ directly affects the penalty payment in round $t$.
    Intuitively, a player is more reluctant to lie if $b_{t+1}$ is high and better off lying if $b_{t+1}$ is small. 
    
    \item [(2)] The number of rounds left to play, i.e., $t$. 
    The value of $t$ indirectly influences the probability and the number of times a player will be busted in remaining rounds. 
\end{itemize}

Via Lemmas~\ref{lem:single-sub-opt}-\ref{lem:single-decr-thres}, we show these are the {\em only two} elements that determine a rational player's action. 
The following lemma shows that it is not optimal for a player to report a value strictly between 0 and $D$.
Moreover, if a player is untruthful in the previous round, it is better off to remain untruthful. 
With this lemma, we largely reduce the strategy space we need to consider.

\begin{lemma} \label{lem:single-sub-opt}
For any round $t \le T$, given $y_t = 0$, the optimal report in round $t$ is $b_t\in \{0,D\}$.
Moreover, if $t < T$ and $b_{t+1} = y_{t} = 0$, then the optimal report is $b_t=0$.
\end{lemma}

Next, we  prove that in each round, the optimal strategy is determined by a penalty threshold such that a player will be truthful if and only if the penalty rate $r$ is above the threshold. 
We call them {\em critical thresholds}.

\begin{lemma}[Critical Thresholds]\label{lem:critical-thresholds}
For $t=T$,  there is a threshold penalty rate $r_T^{(\emptyset)} \ge 0$ such that reporting $D$ is optimal if and only if the penalty rate is at least $r_T^{(\emptyset)}$;
For $t<T$, there is a threshold penalty rate $r_t^{(b_{t+1})} \ge 0$ such that reporting $D$ is optimal for a player in round $t$ with history $b_{t+1}$ if and only if the penalty rate is at least $r_t^{(b_{t+1})}$.
\end{lemma}

Lemmas~\ref{lem:single-sub-opt} and \ref{lem:critical-thresholds} together imply that  the optimal strategy can only be one or a combination of the basic strategies.
In particular, by Lemma~\ref{lem:single-sub-opt}, $r_t^{(0)}=\infty$ for any $t$. 
Moreover, since $b_{t+1}$ can only be $0$ or $D$, by
Lemma~\ref{lem:critical-thresholds}, we only need to determine the values of $r_T^{(\emptyset)}$ and $r_t^{(D)}$ for $t<T$ to complete the picture of optimal strategies. 
In the following two lemmas, we give some properties of these thresholds.

\begin{lemma}\label{fact:single-ber-comp}
$r_t^{(\emptyset)} \ge r_t^{(D)}$ for $t \in \{1,\dots,T\}$.
\end{lemma}

Given the same $t$ rounds left, Lemma~\ref{fact:single-ber-comp} says a player is more inclined to lie without a history than with a truthful history. 
This is straightforward as lying with a truthful history results in an additional penalty payment.

\begin{lemma}\label{lem:single-decr-thres}
Given $r_t^{(\emptyset)}\ge\tfrac{1}{p}$, $r_t^{(\emptyset)}$ decreases as $t$ increases.
\end{lemma}

Lemmas~\ref{fact:single-ber-comp} and \ref{lem:single-decr-thres} together tell us 
the player is least incentivized to be truthful on the first round and 
$r_T^{(\emptyset)}$ is the penalty threshold that ensures  truthfulness for the game.
We give this important threshold in Theorem~\ref{thm:ber-p-no-hist}.


\begin{theorem}\label{thm:ber-p-no-hist}
The minimum penalty for truthful reporting in a game of $T$ rounds with $Ber(p)$ distribution is
    \begin{equation}\label{eq:ber-p-no-hist}
        r_T^{(\emptyset)}=\frac{1-(1-p)^T}{p-p(1-p)^{T-1}}.
    \end{equation}
\end{theorem}
\begin{proof}
By the definitions of the thresholds, if the penalty $r\ge r_T^{(\emptyset)}$ and $r\ge r_t^{(1)}$ for any $t \le T$, then the player will be truthful. 
By Lemma~\ref{fact:single-ber-comp} and \ref{lem:single-decr-thres}, $r_T^{(\emptyset)}\ge r_t^{(\emptyset)}\ge r_t^{(1)}$, for any $t\le T$.
Therefore, it is only necessary to compute $r_T^{(\emptyset)}$.
By Lemma~\ref{lem:single-sub-opt} and \ref{fact:single-ber-comp}, it is sufficient to compare lying-till-busted and honest-till-end in the first segment:
\begin{equation*}
\begin{array}{rll}
 \bbE[\text{honest}] &= D\cdot\bbE[\text{\# days before busted}]\\[5pt]
 & = D+D\left\{\sum_{i=0}^{T-2}i(1-p)^ip+(T-1)(1-p)^{T-1}\right\}\\[5pt]
        &= D+D\cdot\frac{1-p}{p}(1-(1-p)^{T-1}); \\[5pt]
        \bbE[\text{lying}] &= rD\cdot\text{Pr(busted)} = rD(1-(1-p)^{T-1}).
\end{array}
\end{equation*}
    The optimal threshold can be obtained via setting these two expected costs equal,
    $$r_T^{(\emptyset)} = \frac{D+D\frac{1-p}{p}(1-(1-p^{T-1}))}{D(1-(1-p)^{T-1})}=\frac{1-(1-p)^T}{p-p(1-p)^{T-1}}.$$ 
\end{proof}

We see $r_T^{(\emptyset)} \rightarrow 1/p$ as $T \rightarrow \infty$ and $r_T^{(\emptyset)}$ decreases as $T$ increases.
This implies the increasing length of the game incentivizes the player to speak the truth today, even when they do not have to. 
To understand Theorem \ref{thm:ber-p-no-hist}, we observe that it is sufficient to compare lying-till-busted and honest-till-end since $r_T^{(\emptyset)}$ ensures the player to stay truthful after being busted.  
Before 
the player is busted for the first time, it is not optimal to oscillate between lying and truth-telling, as it is strictly dominated by lying completely. 
Therefore, the only viable strategies  are lying-till-busted and honest-till-end, and 
the desired threshold sets the expected cost of these two strategies equal.

With a more involved argument, we get the exact values for the truthful threshold given a truthful history, 
i.e, the $r_t^{(D)}$'s.
The values of $r_T^{(\emptyset)}$ and $r_t^{(D)}$ characterize the optimal strategies for a player and are an alternative representation of Table~\ref{tab:ber:opt-str}.

\begin{theorem} \label{thm:single-thres-with-hist}
    For $p\le \tfrac{1}{2}$, $r_t^{(D)}=\tfrac{1-(1-p)^t}{2p-p(1-p)^{t-1}}$. For $p>\tfrac{1}{2}$, $r_t^{(D)}=1$ for $t=1$ and $r_t^{(D)}=\tfrac{1}{2p}$ for $t\ge 2$.
\end{theorem}
\proof{}
 See Appendix~\ref{app:thm2}. 
\endproof
%


The optimal strategy is visualized in Figs.~\ref{fig:thres-low-p} and \ref{fig:thres-high-p} for $p=0.3$ and $p=0.7$, respectively. 
The $x$-axis is the number of rounds left ($t$), and the $y$-axis is the penalty thresholds for truthfulness.
We give examples of penalties via the red dashed lines.
For the first round, the player refers to the blue dot representing $r_T^{(\emptyset)}$ and is truthful if and only if the penalty is above the blue dot. 
Afterwards, given $t$ rounds left and history $D$, the player looks at the green curve representing $r_t^{(D)}$ and is only truthful if the penalty is above the curve.
If the history is $0$, she remains untruthful and reports $0$.  Figs.~\ref{fig:thres-low-p} and \ref{fig:thres-high-p} visualize the optimal strategies given in Table~\ref{tab:ber:opt-str}.
Both green curves are closely related to $\tfrac{1}{2p}$. 
An intuition is that in any round $t<T$, a player pays $D$ if she is truthful and roughly $2prD$ if she lies, where the penalty payment $rD$ comes from the previous  and the next round, each with probability $p$.
The penalty that sets these two costs equal is $\tfrac{1}{2p}$. 
The actual $r_t^{(D)}$ thresholds vary upon values of $t$ and $p$.

\begin{figure}[ht!]
    \centering
    \includegraphics[width=0.6\textwidth]{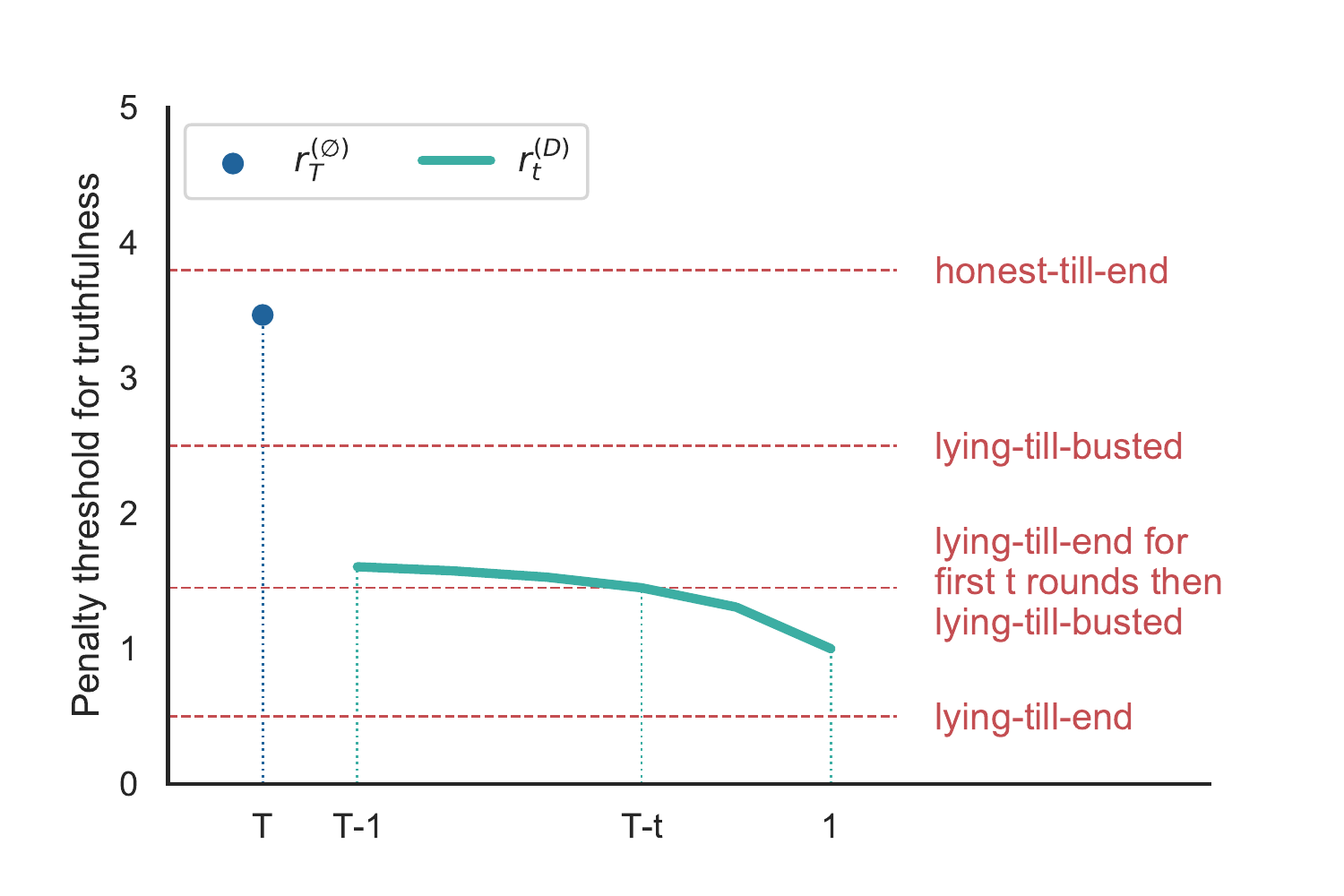}
    \caption{Critical thresholds for each round under $Ber(0.3)$ distribution 
    with examples of  optimal strategies.}
    \label{fig:thres-low-p}
\end{figure}

\begin{figure}[ht!]
    \centering
    \includegraphics[width=0.6\textwidth]{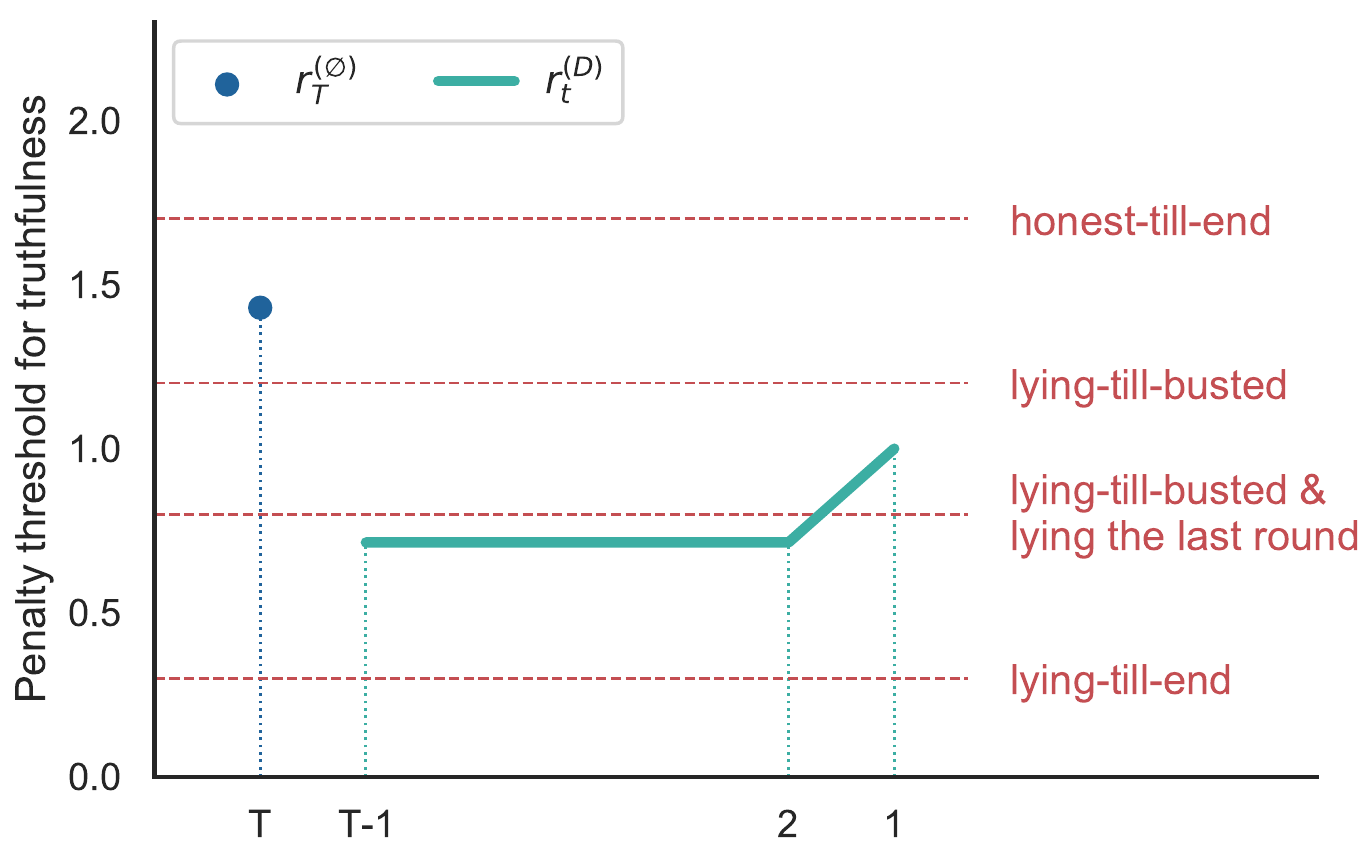}
    \caption{Critical thresholds for each round under $Ber(0.7)$ distribution
   with examples of  optimal strategies.}
    \label{fig:thres-high-p}
\end{figure}


\section{A Reduction for Arbitrary Distributions}\label{sec:single-general}

As discussed in the introduction, only the infinite penalty rate will guarantee complete truthfulness under arbitrary distributions,    
yet there is still hope to obtain approximate results. 
The trick is to redefine being busted  as having a partial signal that is less than $\alpha$ times the true consumption, for $\alpha \in [0,1]$.  
 Then any arbitrary distribution is reduced to $Ber(p)$ where $p$ is the probability that the partial signal is at least $\alpha D$. 

For approximate truthfulness, we define being $\alpha$-truthful as reporting at least $\alpha D$.  
We reuse the arguments of comparing basic strategies from Section~\ref{sec:single} to determine an upper bound for the penalty rate that guarantees $\alpha$-truthfulness.
We introduce the notion of approximate truthfulness in Definition~\ref{def:approx} and give the reduction in Theorem \ref{thm:single-approx}. 
We demonstrate the reduction with uniform distributions in Example~\ref{ex:unif}.

\begin{definition}[$\alpha$-truthfulness]\label{def:approx}
A reporting $\bb$ is $\alpha$-truthful when $b_t\ge \alpha D$ for all $t=1,\dots,T$. 
\end{definition}

\begin{theorem}\label{thm:single-approx}
Given $\alpha \in [0,1]$ and an arbitrary distribution with CDF $F$, 
if a penalty rate $r$ achieves complete truthfulness for $Ber(p)$ where $p=1-F(\alpha D)$,  
then the same $r$ achieves $\alpha$-approximate truthfulness for distribution $F$.\footnote{The reduction
 depends on the 
players' gross consumption, which is private information. 
In reality, if the mechanism has some information on upper bounds of $D$, we are still able to set a penalty rate (which may not be minimum) to obtain truthfulness.}
\end{theorem}
\proof{}
Recall that being ``busted'' means the player has a $D$ realization.
For general distributions, given $\alpha\in[0,1]$, we redefine being busted as having a realization at least $\alpha D$.
Then the probability of being busted is $1-F(\alpha D)$. 
The proof is essentially the same as that of Theorem~\ref{thm:ber-p-no-hist} for $p=1-F(\alpha D)$. We analyze the segment between the first day and the  day when the player is busted.
Assume the minimum report from the player during the segment is $\beta D$, $\beta <\alpha$. 
We  compare the savings the player gets from using this strategy versus reporting $\alpha D$ and the corresponding additional penalty that she needs to pay. 
\begin{equation*}
\begin{array}{lll}
        &\bbE[\text{savings}] \le (\alpha-\beta)D \cdot\bbE[\text{\# days before busted}]\\[4pt]
        &\bbE[\text{penalty}] \ge  rD(\alpha-\beta)\cdot\text{Pr(busted)} 
\end{array}   
\end{equation*}
The player will report $\alpha D$ every round in the segment when expected penalty exceeds expected savings. 
The $(\alpha-\beta)$ term is canceled and the rest calculation is the same as the $Ber(p)$ case where $p=1-F(\alpha D)$.  
\endproof

\begin{example}\label{ex:unif}
Assume partial signals follow a uniform distribution $U(0,D)$.
Let $r$ be the truthful threshold of $Ber(p)$ where $p=1-\alpha$, i.e. $r=\tfrac{1-\alpha^T}{(1-\alpha)(1-\alpha^{T-1})}$. 
Then using $r$ ensures $\alpha$-truthfulness for $U(0,D)$ by Theorem~\ref{thm:single-approx}. 
For uniform distributions, it is impossible to obtain complete truthfulness unless $r=\infty$, 
which can be verified by setting $\alpha = 1$. 
\end{example}

\section{Extension: A Cost Sharing Model} 
\label{sec:multi}

We extend the problem to the multi-player setting and
 focus on the cost sharing among homogeneous players. 
Let $N$ be the set of players with $n = |N| \ge 1$. 
Each player $i\in N$ has a private value $x^{i} \ge 0$, and we assume all players are symmetric, i.e., $x^{i} = D$ for all $i \in N$ (see Appendix~\ref{app:const} for a relaxation).
All players in $N$ split an overhead cost $C$, which is at least the total gross consumption, i.e., $C \ge n D$. 
The game has $T$ rounds in total. 
Given penalty rate $r$, we analyze the following multi-player flux mechanism.




\begin{itemize}
    \item The center observes a partial signal representing player $i$'s net consumption $y_t^{i} \sim F$ for each player $i\in N$;
    \item Each player $i$ submits their reported gross consumption that is at least their net consumption, $b_t^{i} \ge y_t^{i}$;
    
    \item If $t< T$, player $i$'s pays regular payment $C\cdot \tfrac{b_t^{i}}{\sum_j b_t^{j}}$ and penalty payment $r \cdot| b_{t+1}^{i}-b_{t}^{i}|$.
    If $t = T$, the players only pay regular payments.
\end{itemize}

We call 
$b^{i}_{t+1}$ the {\em history} for player $i$ in round $t$ and $\bb_{t+1}$ the {\em group history}. 
If everyone lies in a round, the overhead cost is split evenly among all players.
A mechanism is {\em truthful} if every player reports $D$ for every round.
%
%
We are interested in computing the minimum penalty rates such that truthful reports form a Nash equilibrium (NE) or a dominant strategy equilibrium (DSE).
Informally, a strategy profile is a NE if no player wants to unilaterally deviate, and it is a DSE if no player wants to deviate no matter what the other players do.
We show that 
approximate results for any arbitrary distribution 
can be deducted from an exact analysis for a Bernoulli distribution.
Due to space limit, we  defer all proofs 
to Appendix~\ref{app:multi}.




Fix an arbitrary player $i$ and the other players' strategy $\bsigma^{-i}$.
Let $\bsigma^{-i}_t(\bb_{t+1}, \by_{t}^{-i})$ denote the reported gross consumption by players $j \ne i$ with group history $\bb_{t+1}$ and realizations $\by_t^{-i}$ in round $t$.
A strategy $\bsigma_i$ is called $\bsigma^{-i}$'s {\em best response} if it is the solution of the following recursion. 
\begin{align*}
    OptCost(t, r, \bsigma^{-i}, \bb_{t+1})= \min_{b_t^{i}} ExpCost(t,r, \bsigma^{-i}, \bb_{t+1},b_t^{i}),
\end{align*}

\noindent where the expected cost can be expanded as
\begin{align*} 
    &ExpCost(t,r, \bsigma^{-i}, \bb_{t+1},b_t^{i}) = \\
    &\bE_{\by_t} \Bigg[ \frac{C\cdot \max\{y_t^{i}, b_t^{i}\}}{\sum_{j \ne i}\sigma_t^{j}(\bb_{t+1}, \by_t^{-j}) + \max\{y_t^{i}, b_t^{i}\}} + r \cdot \mid\max\{y_t^{i}, b_t^{i}\} - b_{t+1}^{i}\mid\\
    &\qquad\>\>  + OptCost\left(t-1, r, \bsigma^{-i},(\bm{\sigma}^{-i}(\bm{b}_{t+1},\by_t^{-i}),\max\{y_t^{i}, b_t^{i}\})\right)\Bigg].
\end{align*}

\noindent For the first round when $t=T$, the player would like to minimize the total expected cost, i.e., 
$$OptCost(T,r,\bsigma^{-i}) = \min_{b_T^{i}} ExpCost(T,r, \bsigma^{-i}, b_T^{i}).$$

\noindent Given a strategy profile $\bsigma$,  
if $\bsigma^i$ is a best response to $\bsigma^{-i}$ for every player $i$, then $\bsigma$ is called a Nash equilibrium.
If $\bsigma_i$ is a best response to any $\bsigma'^{-i}$ (not necessarily $\bsigma^{-i}$) for any player $i$, 
$\bsigma_i$ is then called a  dominant strategy equilibrium.

Similar to the single player setting, we avoid solving the recursion by exploiting the properties of the mechanism.
Again, we start our analysis with $F$ being a Bernoulli distribution and provide a reduction for approximate truthfulness when $F$ is an arbitrary distribution. 
In the single player model with Bernoulli-distributed $F$, we have shown that it is only optimal for a player to report $0$ or her actual consumption $D$.
We claim it is the same case for multiple players.
Moreover, if a player lied yesterday and also has an observed consumption of 0 today, they will report 0 regardless of other player's actions.

\begin{lemma}\label{lem:multi-sub-opt}
For Bernoulli-distributed $F$, reporting anything strictly between 0 and D is sub-optimal  in a multi-player flux mechanism.
Moreover, if $b_{t+1}^{i}=y_t^{i}=0$, it is optimal to report $b_t^{i}=0$.
\end{lemma}

Starting from this point, we assume that every player reports either $0$ or $D$.
When $n=2$, we show that the multi-player model reduces to the single player model with a multiplicative factor of $\frac{C}{2D}$.
The reason of the reduction is that the savings of switching to lying from being truthful for a player is always $\frac{C}{2}$, regardless of what the other player does.

\begin{lemma}\label{lem:multi_n=2}
When $n=2$, the multi-player model reduces to single player model.
The truthful penalty threshold is $\frac{C}{2D}$ times (\ref{eq:ber-p-no-hist}). 
\end{lemma}

For general $n$, we show it is sufficient to analyze the maximum difference between lying and truth-telling for player $i$ in round $t$ given group history $\bb_{t+1}$.
In a DSE, a player achieves the biggest gain from lying if all players were lying in the previous round. 
We then use $\bb_{t+1}=\bzero$ to compare lying and truth-telling for a player. 

\begin{theorem}\label{thm:multi-ber-dom}
For the $Ber(p)$ distribution, truthful strategy profile forms a dominant strategy equilibrium if and only if 
\begin{equation}\label{eq:multi-ber-dom-main}
    r\ge \frac{C}{nD}\frac{1-(1-p)^{n-1}}{p}\frac{1-(1-p)^T}{p-p(1-p)^{T-1}}.
\end{equation}
\end{theorem}

If we slowly lower the penalty from (\ref{eq:multi-ber-dom-main}), we will hit a threshold such that truth-telling is an NE. 
The difference between the truthful NE and the DSE is that now we can assume that every player $j \ne i$ is truthful in the first round and show that player $i$ would not deviate unilaterally.
However, we shall not assume that player $j \ne i$ remains truthful for the rest of the game.
This is because if player $i$ lies in the first round, player $j$ can observe the report of $i$ in the second round and deviate from truthful behavior.
We first show that if $r\ge \frac{C}{nD}\frac{1}{p}$, players with truthful history stay truthful. 
Then we can safely assume player $j \ne i$ remains truthful throughout the game.
In this way, truthful NE is reduced to the case where there is one strategic player and $n-1$ truthful players.
It is not hard to see the threshold is precisely $\frac{C}{nD}\frac{1-(1-p)^T}{p-p(1-p)^{T-1}}$.

\begin{theorem}\label{thm:multi-exist}
For the $Ber(p)$ distribution, truthful strategy profile forms a Nash equilibrium if and only if 
\begin{equation}\label{eq:multi-exist-main}
    r\ge \frac{C}{nD}\frac{1-(1-p)^T}{p-p(1-p)^{T-1}}.
\end{equation}
\end{theorem}


We visualize $Ber(p)$ penalty thresholds in Fig.~\ref{fig:thres-multi} for different $T$'s and $p$'s.
The $x$-axis is the total number of rounds for a game and 
the $y$-axis is the penalty rate that guarantees the specified equilibrium. 
The blue and orange lines are penalty thresholds for $p=\tfrac{1}{3}$ and $\tfrac{2}{3}$, respectively. 
The solid and dashed lines are thresholds for truthful DSE and NE, respectively. 
All four thresholds in Fig.~\ref{fig:thres-multi} decrease as $T$ increases, suggesting that the increasing length of the game promotes truthful equilibria. 
From expressions (\ref{eq:multi-ber-dom-main}) and (\ref{eq:multi-exist-main}), we see that the DSE and NE thresholds tend to be the same as $p$ approaches 1. 

\begin{figure}[htbp]
    \centering
    \includegraphics[width=0.6\textwidth]{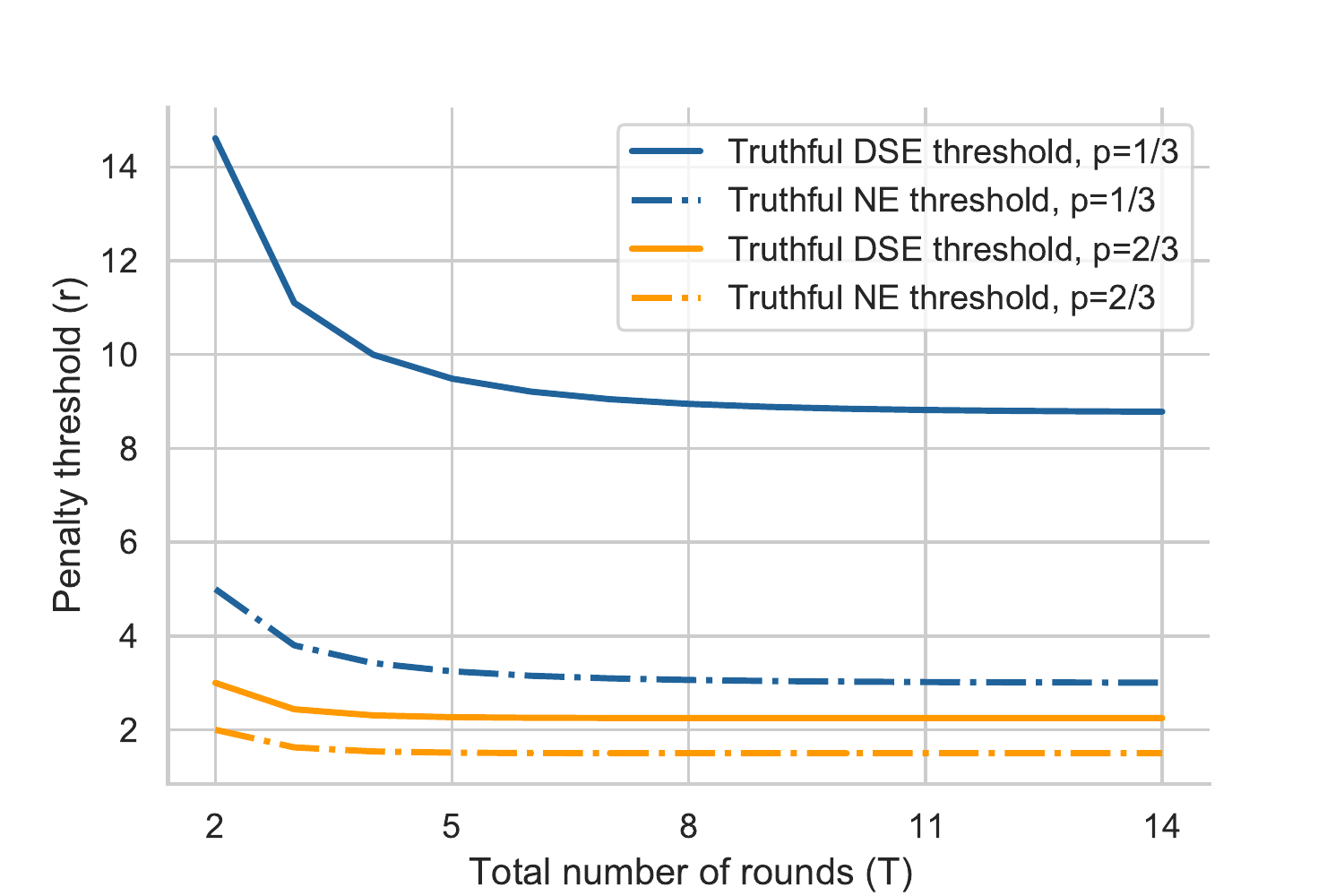}
    \caption{Exact penalty thresholds for truthful DSE and NE, given total number of rounds $T$ for $Ber(p)$ distributions. We assume $n=20$, $D=1$ and $C=n\cdot D = 20$.}
    \label{fig:thres-multi}
\end{figure}

Similar to the single-player model, we  extend the results for Bernoulli distributions to approximate results for  general distributions.
Given $\alpha\in[0,1]$, we redefine being busted as having an observed consumption at least $\alpha D$. 
For the dominant strategy equilibrium, we find the threshold such that being $\alpha$-truthful is a dominant strategy.
For Nash equilibrium, we first define the approximate truthful NE, which is a natural extension of the complete truthful NE.

\begin{theorem}\label{thm:multi-gen-dom}
Given $\alpha\in[0,1]$ and some general distribution $F$, 
let $p=1-F(\alpha D)$.
The $\alpha$-truthful strategy profile forms a dominant strategy equilibrium if
\begin{equation}\label{eq:multi-gen-dom}
    r\ge \frac{1}{\alpha}\frac{C}{nD}\frac{1-(1-p)^{n}}{p}\frac{1-(1-p)^T}{p-p(1-p)^{T-1}}.
\end{equation}
\end{theorem}

\begin{definition}[$\alpha$-truthful Nash equilibrium]
Given $\alpha \in [0,1]$, a reporting profile $\bm{b} \in [0,D]^{n\times T}$ is an $\alpha$-truthful Nash equilibrium if $\bm{b}_{t}^{i} \ge \alpha D$ for all $i, t$ and no player wants to deviate from being $\alpha$-truthful in any round.
\end{definition}

\begin{theorem}\label{thm:multi-gen-ne}
Given $\alpha\in[0,1]$ and some general distribution $F$,
let $p=1-F(\alpha D)$.
The $\alpha$-truthful strategy profile forms a Nash equilibrium if
\begin{equation}\label{eq:multi-gen-ne}
    r\ge \frac{1}{\alpha}\frac{C}{nD}\frac{1-(1-p)^T}{p-p(1-p)^{T-1}}.
\end{equation}
\end{theorem}

We see that both the penalty thresholds, \eqref{eq:multi-gen-dom} and \eqref{eq:multi-gen-ne} are close to $\frac{1}{\alpha}$ times their Bernoulli thresholds, \eqref{eq:multi-ber-dom-main} and \eqref{eq:multi-exist-main}, for  $p = 1-F(\alpha D)$. 
{Recall that in the single-player model, $\alpha$-truthfulness can be obtained by directly using the Bernoulli threshold with $p = 1-F(\alpha D)$.
In the multi-player model, however, we have to multiply the Bernoulli threshold with a factor of $\frac{1}{\alpha}$, which suggests it is more difficult to get every player to speak the truth under the cost sharing setting.}  
We note that both penalty rates (\ref{eq:multi-gen-dom}) and (\ref{eq:multi-gen-ne}) are upper bounds for the actual thresholds.
This is because we treat any report greater than $\alpha D$ as $\alpha D$. 
We conjecture that the exact thresholds are not far from thresholds (\ref{eq:multi-gen-dom}) and (\ref{eq:multi-gen-ne}).

\section{Conclusion and Open Problems}\label{sec:open-problem}

We propose a penalty mechanism for eliciting truthful self-reports when only partial signals are revealed in a repeated game. 
A player faces trade-off between under-reporting today and paying a penalty in the future due to the uncertainty of partial signals. 
We find that the length of the game naturally reduces the minimum penalty rate that incentivizes truth-telling.
Given any penalty rate, we give a 
characterization 
of the optimal strategies under both single- and multiple-player settings for any distribution.
We identify a penalty rate that achieves complete truthfulness for Bernoulli distributions, which can be used in a reduction to obtain approximate truthfulness for arbitrary distributions.  
 
A possible future direction is to extend 
our results 
to asymmetric multi-player settings where players do not have the same gross consumption or the same distribution for partial signals.
For heterogeneous players, we may then consider, in addition to truthfulness,  the fairness of the mechanism.
It would be interesting to develop a definition of fairness for the cost sharing model and compute the fairness ratios accordingly. 
It is also worthwhile to derive other truthful and fair mechanisms that do not involve penalty.  

\newpage

\bibliography{references}
\bibliographystyle{alpha}

\newpage

\appendix
\section*{Appendix: Missing Materials}
\section{Justifying the Assumption that Gross Consumption is Constant}\label{app:const}
In the prosumer pricing problem, we adopt the assumption from Khodabakhsh~\etal \cite{khodabakhsh2019prosumer} that $D$, the gross electricity usage, remains constant throughout the time horizon. 
According to U.S. Energy Information Administration, the electricity usage typically follows a daily pattern, which means within some period (e.g., day, month or season), the electricity consumption does not vary much \cite{eia}. 
This is  especially true for industrial sites, which are the major consumers for utilities \cite{nadel2014electricity}. 
Therefore, we can always discretize the time horizon into sub-intervals such that the electricity consumption within each interval is relatively constant.

To relax this assumption, let the gross consumption for each round come from a known range, i.e., $D_t \in[\underline{D},\overline{D}]$.  
The center knows the $\underline{D}$ and $\overline{D}$ but does not necessarily know $D_t$ for any $t$. 
We show that our results extend straightforwardly. 
Recall that in the analysis of Bernoulli distributions, we compare the basic two strategies and find the penalty rate that sets the two expected costs equal. 
We can find an upper bound for the truthful threshold by bounding the expected cost of lying-till-busted and honest-till-end from below and above, respectively. 
Then the resulting penalty rate is simply the original threshold (\ref{eq:ber-p-no-hist}) times the ratio $\overline{D}/\underline{D}$. 
For the reduction from Bernoulli distributions to any arbitrary distribution with cdf $F$, given $\alpha\in[0,1]$, we now define $p=1-F(\alpha\overline{D})$ and use the same argument to obtain an upper bound of penalty rate that achieves $\alpha$-truthfulness. 
For the multi-player model, we add the multiplicative ratio $\overline{D}/\underline{D}$ in every expression for an upper bound of the desired penalty rate. 
This relaxation will also help add heterogeneity to the multi-player model.
\section{Missing Proofs in Sections~\ref{sec:single}}\label{app:pfs}


\subsection{Proof for Lemma~\ref{lem:single-sub-opt}}

\proof{}
    To see the first sentence, we can observe that the cost function is a linear function of today's report $b_t$ and thus either $0$ or $D$ achieves the optimality. 
    To see the second sentence, we consider the last round $t$ in the optimal strategy such that when $(b_{t+1}, y_t) = (0,0)$ but $b_t > 0$. 
    It is obvious if $t$ is the last round, and thus we assume $t>1$.
    By reporting $b_t$ in round $t$, the expected total cost afterward is
    \begin{align*}
     & rb_t + b_t + ExpCost(t-1,r,b_t, b_{t-1}) \nonumber\\
        &= \bE_{y_{t-1}}[(r+1)b_t + r|\max\{y_{t-1}, b_{t-1}\} - b_{t-1}|  + \max\{y_{t-1}, b_{t-1}\}\\
        &\qquad\qquad + OptCost(t-2, r, \max\{y_{t-1}, b_{t-1}\})] \nonumber\\
        &> \bE_{y_{t-1}}[ r\max\{y_{t-1}, b_{t-1}\} + \max\{y_{t-1}, b_{t-1}\}  + Opt(t-2, r, \max\{y_{t-1}, b_{t-1}\})] \nonumber \\
        &= ExpCost(t-1,r,0,b_{t-1})
    \end{align*}
    where the inequality is because for any $x\ge 0$ and $b_t>0$,
    \[
    (r+1)b_t + r|x - b_{t}| \ge (r+1)b_t + r(x- b_{t}) > r x.
    \]
    The last term 
    is exactly the expected total cost by reporting 0 in round $t$ but adopting the same strategy with the optimal one afterward, which is contradiction with $b_{t}>0$ being optimal. 
    Thus we complete the proof of the lemma.
\endproof

\subsection{Proof for Lemma~\ref{lem:critical-thresholds}}

\proof{}
Note that by Lemma~\ref{lem:single-sub-opt}, $b_{t+1}$ can only be $0$ or $D$.
Moreover, $r_t^{(\emptyset)}=\infty$ for any $t<T$.
Therefore, we only need to show the existence of $r_T^{(\emptyset)}$ and $r_t^{(D)}$ for $t<T$.
It suffices to show the following claim:
For any round $T \ge t \ge 1$ with $y_t = 0$, if the optimal strategy is $b_t = 0$ given penalty rate $r$, then  $b_t = 0$ is also optimal for any $r' \le r$; if the optimal strategy is $b_t = D$ given penalty rate $r$, then  $b_t = D$ is also optimal for any $r' \ge r$.

To prove the claim, we use induction on $t$.
When $t=1$, it is easy to see that $r_1^{(D)}$ exists and is equal to 1.
Consider $t>1$ rounds left and the optimal strategy is to report $0$ given penalty $r$ and $b_{t+1}=D$ (or no history if $t=T$).
If we increase penalty $r$ to $r'>r$, by induction, the optimal strategy for future rounds either remains the same or switch to $D$ from $0$.
Given history $D$ and report $D$, the payment for thin round is independent from penalty rate $r$.
Therefore, 
reporting $D$ is still optimal.
A similar argument can be made for reporting $0$. 
\endproof

\subsection{Proof for Lemma~\ref{fact:single-ber-comp}}

\proof{}
Given the same $t$ rounds left, it is straightforward to see that players with a truthful history is more incentivized to lie compared to when she has no history.
This is because she needs to pay an additional payment of $rD$ whenever she has a truthful history.
Mathematically, let $r\ge r_t^{(\emptyset)}$.
Then we have 
\begin{align*}
    &ExpCost(t,r,D) \le ExpCost(t,r,0) \\
    \implies &D+OptCost(t-1,r,D)\le p(D+OptCost(t-1,r,D))+(1-p)\cdot OptCost(t-1,r,0)\\
    \implies &D+OptCost(t-1,r,D)\le OptCost(t-1,r,0)\\
    \implies &D+OptCost(t-1,r,D)\le rD+OptCost(t-1,r,0)\\
    \implies &ExpCost(t,r,D,D)\le ExpCost(t,r,D,0)
\end{align*}
The above inequalities show that when $r\ge r_t^{(\emptyset)}$, the player prefers truth-telling over lying when she has a truthful history, which implies $r_t^{(D)}\le r_t^{(\emptyset)}$. 
\endproof

\subsection{Proof for Lemma~\ref{lem:single-decr-thres}}

\proof{}
An equivalent statement of Lemma~\ref{lem:single-decr-thres} is that given a penalty $r\ge\frac{1}{p}$, if a player is truthful when there are $T$ rounds left, then she is also truthful when there are $T+1$ rounds left.
Let us prove the alternate statement.

Assume $r\ge \frac{1}{p}$ and that the player is truthful when there are $T$ rounds left. 
Let $OPT(*,T)$ denote the optimal cost for a player if she reports $*$ in the current round and there are $T$ rounds left.
Since the player is truthful when there are $T$ rounds left, we have 
$$D+OPT(D,T-1)\le OPT(0,T-1).$$
Now assume there are $T+1$ rounds left and the player is free to lie in the first round. By Lemma~\ref{lem:single-sub-opt}, there are the following two cases.
\begin{itemize}
    \item player reports $D$ in the first round
\end{itemize}    
 If the player also reports $D$ in the second round, she pays $2D+OPT(D,T-1)$.
Otherwise she pays $D+rD+OPT(0,T-1)\ge 2D+rD+OPT(D,T-1)$, which is dominated by reporting $D$ for both rounds.   
 \begin{itemize}   
    \item player reports $0$ in the first round
\end{itemize}
    Then the player's total expected payment is
\begin{align*}
   &p(D+rD+OPT(D,T-1))+(1-p)OPT(0,T-1)\\
   &\ge D+prD+OPT(D,T-1)\ge 2D+OPT(D,T-1).
\end{align*}
Therefore, the optimal strategy is to report $D$ in the first two rounds and the rest of the game is exactly the same as when there are $T$ rounds left. 
\endproof


\subsection{Proof for Theorem~\ref{thm:single-thres-with-hist}}\label{app:thm2}
\proof{}
We first show the proof for $p\le \tfrac{1}{2}$ by induction on $t$.
Let $h(t)=\tfrac{1-(1-p)^{t}}{2p-p(1-p)^{t-1}}$.
%
%
Assume there are $t$ rounds left.
Note that $h(t)$ increases in $t$, which means that if $r\ge h(t)$, then $r \ge h(t')$ for $t'\le t$, i.e., the player stays truthful for the rest of the $t$ rounds.
Similar to the argument in Theorem~\ref{thm:ber-p-no-hist}, we compare the expected payments of the two strategies, namely lying-till-busted (``lying") and being honest, within a segment.
Note that the segment now starts with being busted, because the player has a truthful history. 
\begin{align*}
        \bbE[\text{honest}] &= D\cdot\bbE[\text{\# days before busted}]= D+D\cdot\tfrac{1-p}{p}(1-(1-p)^{t-1}); \\
        \bbE[\text{lying}] &=rD+ rD\cdot\text{Pr(busted)} = rD(2-(1-p)^{t-1}).
\end{align*}
The penalty that results in truthfulness sets these two payments equal, i.e. $r=\tfrac{1-(1-p)^{t}}{2p-p(1-p)^{t-1}}=h(t)$.

The proof for $p\ge\tfrac{1}{2}$ is slightly different.
First note that for $t=1$, it is not hard to see the threshold $r_1^{(D)}=1$ by comparing the cost of being honest (i.e., $D$) and the cost of lying (i.e., $rD$).
For $t>1$ rounds left, we apply the same argument above, with the consideration that the player will switch to lying in the very last round if she is allowed to.
Therefore, we have 
\begin{align*}
        \bbE[\text{honest}] &= D\cdot\bbE[\text{\# days before busted}]-(1-r)D\cdot\text{Pr(not busted in the last day)}\\
        &= D+D\cdot\tfrac{1-p}{p}(1-(1-p)^{t-1})-(1-r)D(1-p)^{t-1}; \\
        \bbE[\text{lying}] &=rD+ rD\cdot\text{Pr(busted)} = rD(2-(1-p)^{t-1}).
\end{align*}
The penalty that sets the above two expected costs equal is $\tfrac{1}{2p}$.  
\endproof

%
%


\section{Missing Proofs in Section~\ref{sec:multi}}\label{app:multi}


\subsection{Proof for Lemma~\ref{lem:multi-sub-opt}}



\proof{}
We can use a similar argument in the proof for Lemma~\ref{lem:single-sub-opt} to prove that if a player lied yesterday, it is better off to lie today.
We consider the last round $t$ in the optimal strategy such that when $(b^{i}_{t+1}, y^{i}_t) = (0,0)$ but $b^{i}_t > 0$. 
    It is obvious if $t$ is the last round, and thus we assume $t>1$.
    By reporting $b^{i}_t$ in round $t$, the expected total cost afterward is
    
\begin{align*}
        & rb^{i}_t + \frac{C\cdot b_t^{i}}{\sum_{j}b_t^{j}} + ExpCost(t-1,r,\bsigma^{-i},\bb_t, b^{i}_{t-1}) \nonumber\\
        &= \bE_{y_{t-1}}\Bigg[\left(r+ \frac{C}{\sum_{j}b_t^{j}}\right)b^{i}_t + r|\max\{y^{i}_{t-1}, b^{i}_{t-1}\} - b^{i}_{t}| \\
        &\qquad\qquad\quad + \frac{C\cdot \max\{y_{t-1}^{i}, b_{t-1}^{i}\}}{\sum_{j \ne i}\sigma_t^{j}(\bb_{t}, \by_{t-1}^{-j}) + \max\{y_{t-1}^{i}, b_{t-1}^{i}\}}\\
        &\qquad\qquad\quad  + OptCost\Big(t-2, r, \bsigma^{-i},(\bm{\sigma}^{-i}(\bm{b}_{t},\by_{t-1}^{-i}),\max\{y_{t-1}^{i}, b_{t-1}^{i}\})\Big)\Bigg] \nonumber\\
        &> \bE_{y_{t-1}}\Bigg[ r\cdot\max\{y_{t-1}, b_{t-1}\} + \frac{C\cdot \max\{y_{t-1}^{i}, b_{t-1}^{i}\}}{\sum_{j \ne i}\sigma_t^{j}(\bb_{t}, \by_{t-1}^{-j}) + \max\{y_{t-1}^{i}, b_{t-1}^{i}\}} \\
        &\qquad\qquad\quad + OptCost\Big(t-2, r, \bsigma^{-i},(\bm{\sigma}^{-i}(\bm{b}_{t},\by_{t-1}^{-i}),\max\{y_{t-1}^{i}, b_{t-1}^{i}\})\Big)\Bigg] \nonumber \\
        &= ExpCost(t-1,r,\bsigma^{-i},(\bb_t^{-i},0),b^{i}_{t-1}),
    \end{align*}
which is the expected total cost by reporting 0 in round $t$ but adopting the same strategy with the optimal one afterward.
    This contradicts that $b^{i}_{t}>0$ is optimal. 


To see that partial reporting is optimal, rewrite the payment for the current round as
$$C\left(1-\frac{\sum_{j\ne i}b_{t}^{i}}{\sum_{j\ne i}b_t^{j}+b_t^{i}}\right)+r\mid b_{t+1}^{i}-b_{t}^{i}\mid,$$
whose second derivative is negative with respect to $b_t^{i}$.
This means that the payment function is concave in $b_t^{i}$ and will take minimum at either of the endpoints $0$ and $D$.
\endproof



\subsection{Proof for Lemma~\ref{lem:multi_n=2}}
\proof{}
In the single player model, if a player switches to lying from being honest, she saves $D$ for regular payment and then pays penalty $rD$ if she has a truthful history. 
Now in the two player model, since players are symmetric, we fix the action of player 2 and see what happens with player 1.

\begin{table}[htbp]
    \centering
    \begin{tabular}{c|c|c|c|}
\multicolumn{2}{c}{} & \multicolumn{2}{c}{player 2}\\\cline{3-4}
\multicolumn{2}{c|}{} & Honest & Lying\\\cline{2-4}
 \multirow{2}{*}{player 1}&Honest & $(C/2,C/2)$ & $(C,0)$ \\\cline{2-4}
&Lying & $(0,C)$ & $(C/2,C/2)$ \\\cline{2-4}
\end{tabular} 
    \caption{Expected payment in every round for each player in the multi-player model with $n=2$.}
    \label{tab:exp_pay}
\end{table}
No matter if player 2 is honest or lying, for player 1, switching to lying would save $C/2$ and may cost a penalty payment of $rD$.
By applying the same argument seen in Section~\ref{sec:single-ber} with the new expected savings and penalties, we get the same penalty threshold, except with a $C/2D$ multiplicative factor. 
\endproof

\subsection{Proofs for Theorem~\ref{thm:multi-ber-dom} and \ref{thm:multi-exist}}\label{app:eq-ber}

For general $n$ strategic players, we develop an alternative way to compute the penalty thresholds for NE and DSE.
Interestingly, we only need to make use of the following important definition, $\Delta EC^{i}_{t}(\bm{b}_{t+1})$, to 
derive a universal framework for equilibrium proofs.



\begin{definition}\label{def:delta_ec}
Let $EC^{i}_{t}(\bm{b}_{t+1})$ denote the expected cost for player $i$ with when there are $t$ rounds left and the group history is $\bm{b}_{t+1}$.
Define 
\begin{align*}
   \Delta &EC^{i}_{t}(\bm{b}^{-i}_{t+1}) \triangleq EC^{i}_{t}(b_{t+1}^{i}=D,\bm{b}^{-i}_{t+1}) - EC^{i}_{t}(b_{t+1}^{i}=0,\bm{b}^{-i}_{t+1})
\end{align*}
as the difference in the expected payments by reporting $D$ versus $0$ for player $i$, given $t$ rounds left and the reports of other players, $\bm{b}_{t+1}^{-i}$.  
\end{definition}
To simplify the notation, we remove the superscript $i$ in the definition and write $\Delta EC_{t}(\bm{b}^{-i}_{t+1})$. 
By Lemma~\ref{lem:multi-sub-opt}, $\bm{b}_{t+1}^{-i}$ is a string of size $n-1$ consisting of $0$'s and $D$'s.
We first present a technique to obtain upper bounds of $\Delta EC_t(\bm{b}_{t+1}^{-i})$ given $\bm{b}_{t+1}^{-i}$.

\begin{lemma}\label{lem:delta-ec-exp}
Some upper bounds of $\Delta EC_t(\bm{b}_{t+1}^{-i})$:
\begin{itemize}
    \item[(i)] $\Delta EC_t(\bzero)$
    $$\le \frac{C}{n}\frac{1-(1-p)^{n-1}}{p}\sum_{i=1}^{t}(1-p)^i-prD\sum_{i=0}^{t-1}(1-p)^i$$
    \item[(ii)] $\Delta EC_t(b_{t+1}^{j}=0,\bm{b}_{t+1}^{-i,j}=\bD)$
    $$\le \frac{C}{n-1}\sum_{i=1}^{t}(1-p)^i-prD\sum_{i=0}^{t-1}(1-p)^i$$
\end{itemize}
\end{lemma}
\proof{}
We prove (i) where $\bm{b}_{t+1}^{-i}=\bzero$ and the proof for (ii) is similar.
Let $M=\frac{C}{n}\frac{1-(1-p)^{n-1}}{p}$. 
We prove by induction.

\noindent \underline{Base case}. $t=1$. 
With probability $p$, having a $D$ or $0$ history pays the same regular payment and the $0$ history needs to pay penalty. 
With probability $1-p$, only the $D$ history pays the regular payment. 
\begin{align*}
    \Delta EC_1=EC_1(D)-EC_1(0)&=(1-p)\sum_{k=0}^{n-1}\binom{n-1}{k}p^k(1-p)^{n-1-k}\frac{C}{k+1} -(1-p)^{n}\frac{C}{n}-prD\\
    &=(1-p)\frac{C}{n}\frac{1-(1-p)^{n-1}}{p}-prD\\
    &= (1-p)M-prD.
\end{align*}
Note that $k$ in the second equality represents the number of players being busted in $N\setminus \{i\}$.

\noindent \underline{Induction step}. 
Assume Lemma~\ref{lem:delta-ec-exp} is true for $\Delta EC_t$. 
Consider $t+1$ rounds left.

\begin{align*}
    \Delta EC_{t+1}&=EC_{t+1}(D)-EC_{t+1}(0)\\
    &= (1-p)\sum_{k=0}^{n-1}\binom{n-1}{k}p^k(1-p)^{n-1-k}\left\{\frac{C}{k+1}+EC_{t}(D)\right\}\\
    &\qquad - (1-p)\sum_{k=0}^{n-1}\binom{n-1}{k}p^k(1-p)^{n-1-k}\cdot EC_{t}(0)\\
    &\qquad -prD-(1-p)^n\frac{C}{n}\\
    &\le(1-p)\left\{M+\Delta EC_{t}\right\}-prD\\
    &\le M\sum_{i=1}^{t+1}(1-p)^i-prD\sum_{i=0}^t (1-p)^i. 
\end{align*}
\endproof

An important property of $\Delta EC_{t}(\bm{b}^{-i}_{t+1})$ is that it is monotone increasing as the number of $0$'s in $\bm{b}_{t+1}^{-i}$ increases.
One way to understand this property is that a player $j \ne i$ with a zero history is more likely to lie in the next rounds, which in turn increases the expected regular payment if player $i$ is truthful.
We prove this property mathematically in Lemma~\ref{lem:ec-mnt}.

\begin{lemma}\label{lem:ec-mnt}
If $\hat{\bm{b}}^{-i}_{t+1}$ contains more zeros than $\bm{b}^{-i}_{t+1}$, then 
$$\Delta EC_{t}(\bm{b}^{-i}_{t+1})\le \Delta EC_{t}(\hat{\bm{b}}^{-i}_{t+1}).$$
\end{lemma}
\proof{}
First note that the only non-trivial case is when the penalty is just high enough such that players with truthful history stay truthful and players with 0 history lie whenever realization is 0.
Since every player is symmetric, players with the same history will act the same. 
If the penalty is too low, $\Delta EC_t(\bm{b}_{t+1}^{-i})$ does not depend on $\bm{b}_{t+1}^{-i}$ and $\Delta EC_t(\bm{b}_{t+1}^{-i})-\Delta EC_t(\hat{\bm{b}}_{t+1}^{-i})=0$.
Same when the penalty is too high then players will be truthful regardless of history.
Now we can assume players with truthful history stay truthful regardless of the realization and players with zero history lie whenever possible.
We prove by induction on $t$. 

\noindent \underline{Base case}. $t=1$. Let $\Delta EC_t(\bm{b}_{t+1}^{-i})$ contain $k$ zero's (and $n-k-1$ $D$'s). 
Then we have
\begin{align*}
    \Delta EC_1(\bm{b}_{t+1}^{-i})&= p(-rD) +(1-p)\sum_{i=0}^{k} \binom{k}{i} p^i(1-p)^{k-i} \frac{C}{n-k+i} -\mathbbm{1}_{\{k=n-1\}}\cdot\frac{C}{n}(1-p)^n\\
    &=p(-rD) +(1-p) \sum_{j=0}^{n-1} \alpha(j,k) H(j) \quad \text{letting $j=k-i$}
\end{align*}
where 
\begin{align*}
    \alpha(j,k) = \begin{cases}
    \binom{k}{k-j}p^{k-j}(1-p)^{j} & \mbox{ for } 0 \le j \le k \\
    0 & \mbox{ for } k<j\le n-1
    \end{cases}
\end{align*}
and 
\begin{align*}
    H(j) = \begin{cases}
    \frac{1}{n-j}\cdot C & \mbox{ for } 0 \le j < n-1 \\
    \frac{n-1}{n}\cdot C & \mbox{ for } j= n-1
    \end{cases}
\end{align*}
Note that $\sum_{j=0}^{n-1}\alpha(j,k)=1$ and $\alpha(j,k)$'s depend on $k$. 
On the other hand, $H(j)$'s do not depend on $k$ and is an increasing sequence in $j$. 
Now consider $\hat{\bm{b}}_{t+1}^{-i}$ that contains $\hat{k}$ zeros, and $k<\hat{k}$.
Then we have

\begin{align*}
    \Delta EC_1(\hat{\bm{b}}_{t+1}^{-i})-\Delta EC_1(\bm{b}_{t+1}^{-i}) &= (1-p)\sum_{j=0}^{n-1}\left\{\alpha(j,\hat{k})-\alpha(j,k)\right\} H(j)\\
    &= (1-p)\left\{\sum_{j=k+1}^{\hat{k}}\alpha(j,\hat{k}) H(j)-\sum_{j=0}^{k}(\alpha(j,k)-\alpha(j,\hat{k}))H(j)\right\}\\
    &\ge (1-p)\left\{\sum_{j=k+1}^{\hat{k}}\alpha(j,\hat{k}) H(k)-\sum_{j=0}^{k}(\alpha(j,k)-\alpha(j,\hat{k}))H(k)\right\}\\
    &=(1-p)H(k)\left\{\sum_{j=0}^{\hat{k}}\alpha(j,\hat{k})-\sum_{j=0}^{k}\alpha(j,k)\right\} \\
    &=0
\end{align*}

\noindent \underline{Induction step}. Assume the lemma is true for $t$. We prove for $t+1$ rounds left. 
Assume again $\bm{b}_{t+1}^{-i}$ contains $k$ zero's.
\begin{align*}
    \Delta EC_{t+1}(\bm{b}_{t+1}^{-i})&= EC_{t+1}(D,\bm{b}_{t+1}^{-i}) - EC_{t+1}(0,\bm{b}_{t+1}^{-i})\\
    &= (1-p)\left\{\sum_{i=0}^{k}\binom{k}{i}p^i(1-p)^{k-i}\left(\frac{C}{n-k+i}+\Delta EC_t(k-i \mbox{ lying})\right)\right\}\\
    &\quad -prD-\mathbbm{1}_{\{k=n-1\}}(1-p)^n\left\{\frac{C}{n}- EC_t(0,0)\right\}\\
    &=-prD+(1-p)\sum_{j=0}^{n-1}\alpha(j,k) H(j)
\end{align*}
where $\alpha(j,k)$'s are the same as earlier, and $H(j)$'s are now
\begin{align*}
H(j) = \begin{cases}
\frac{1}{n-j}\cdot C+\Delta EC_t(j \mbox{ lying}) & 0\le j < n-1 \\
\frac{n-1}{n}\cdot C+\Delta EC_t(n-1 \mbox{ lying}) &  j=n-1
\end{cases}
\end{align*}
By induction, $\Delta EC_t(j \mbox{ lying})$ increases in $j$. 
Thus, $H(j)$'s is again an increasing sequence in $j$. 
We re-use the argument in the base case and prove $\Delta EC_{t+1}(\hat{\bm{b}}_{t+1}^{-i}) \ge \Delta EC_{t+1}(\bm{b}_{t+1}^{-i})$ for $\hat{\bm{b}}_{t+1}^{-i}$ with $\hat{k}>k$ zeros.
\endproof

With this property, we develop a framework for the equilibrium proofs of both DSE and NE:
\begin{enumerate}
    \item Determine what $\bm{b}_{t+1}^{-i}$ look like based on the type of the equilibrium we are trying to compute;
    \item Upper bound $\Delta EC_{t}(\bm{b}^{-i}_{t+1})$ with an expression using $C$, $D$, $t$, $p$ and  $r$ (see Lemma~\ref{lem:delta-ec-exp});
    \item Compare player $i$'s expected payment on the first round when she lies or tells the truth using $\Delta EC_{T-1}(\bm{b}^{-i}_{T})$;
    \item Find the penalty rate that sets the two expected payments equal, and that is the desired penalty threshold.
\end{enumerate}



\subsection*{Proof for Theorem~\ref{thm:multi-ber-dom}}\label{sec:multi-ber-dom}

\proof{}
Fix a player $i$.
To show that being truthful is a dominant strategy for player $i$, we want to look at the situation that maximizes the difference between truth-telling and lying for player $i$, which is precisely when every other player is lying as much as possible, by Lemma~\ref{lem:ec-mnt}. 
Now we assume every other player reports $0$ whenever they can. 
We compare the expected cost of being truthful and lying on the very first round.
\begin{align*}
&\bbE[\text{lying}] = (1-p)^{n-1}\frac{C}{n}+ EC_{T-1}(0,\bzero);\\
        &\bbE[\text{honest}]= \sum_{k=0}^{n-1}\binom{n-1}{k}p^i(1-p)^{n-1-k}\left\{\frac{C}{k+1}+EC_{T-1}(D,\bzero)\right\},
\end{align*}
where $k$ represents the number of players in $N\setminus \{i\}$ that are busted in round $T$.  
We would like to find the penalty rate such that $\bbE[\text{honest}]-\bbE[\text{lying}]\le0$.
By Lemma~\ref{lem:delta-ec-exp}, we have
\begin{align*}
    \bbE[\text{honest}]-\bbE[\text{lying}] &= \frac{C}{n}\frac{1-(1-p)^n}{p}-(1-p)^{n-1}\frac{C}{n}+\Delta EC_{T-1}(\bzero)\\
    &\le \frac{C}{n}\frac{1-(1-p)^{n-1}}{p}\frac{1-(1-p)^T}{p}-rD(1-(1-p)^{T-1}),
\end{align*}
which is negative when $r\ge \frac{C}{nD}\frac{1-(1-p)^{n-1}}{p}\frac{1-(1-p)^T}{p-p(1-p)^{T-1}}$.
Since we are analyzing the case that maximizes the differences in lying and truth-telling, we can say that truthfulness is a Nash equilibrium if and only if the penalty rate is above the given threshold.
\endproof


\subsection*{Proof for Theorem~\ref{thm:multi-exist}}

\proof{}
Based on the discussion, we first assume that every player $j\ne i$ is truthful in the first round and $r\ge \frac{C}{nD}\frac{1}{p}$. 
We want to prove that some player $j \in N\setminus \{i\}$ does not want to deviate from being truthful in the next round.
Then it follows that the threshold for truthful NE is equivalent to the case with single sophisticated player and $n-1$ truthful players.  
Since the threshold (\ref{eq:multi-exist-main}) is exact in the model with one sophisticated and $n-1$ truthful players, this threshold is the exact threshold for truthful Nash equilibrium. 

Fix some player $j \ne i$. 
Assume there are $t+1$ rounds left.
Again, we compare the expected payments of lying and being honest for player $j$.
\begin{align*}
    &\bbE[\text{honest}] = p\left\{\frac{C}{n}+EC_{t}(D,\bD)\right\}+(1-p)\left\{\frac{C}{n-1}+EC_{t}(D,(0,\bD))\right\}\\
    &\bbE[\text{lying}] = rD+p\cdot EC_{t}(0,\bD)+(1-p)\cdot EC_{t}(D,(0,\bD))
\end{align*}
By Lemma~\ref{lem:ec-mnt} and Lemma~\ref{lem:delta-ec-exp}, we have 
\begin{align*}
    \bbE[\text{honest}]-\bbE[\text{lying}] \le \frac{C}{n-1}-rD+\Delta EC_t(0,\bD) \le 0,
\end{align*}
for $r\ge \frac{C}{nD}\frac{1}{p}$.
Thus, player $j$ will not deviate from being truthful, even when player $i$ is lying in the previous round. 
\endproof

\subsection{Proof for Theorem~\ref{thm:multi-gen-dom}}
\proof{}
Let $p=1-F(\alpha D)$. 
Assume, for contradiction, that the player adopts some strategy that has a minimum reporting of $\beta D$, $0\le \beta \le \alpha$.
We compare the expected costs of this strategy and the strategy of being $\alpha$-truthful.
We re-define $\Delta EC_t(\bm{b}_{t+1}^{-i})$ as follows:
$$\Delta EC_t(\bm{b}_{t+1}^{-i}) \triangleq EC_t(\alpha D,\bm{b}_{t+1}^{-i})-EC_t(\beta D,\bm{b}_{t+1}^{-i}).$$
Similar to the proof in Theorem~\ref{thm:multi-ber-dom}, we want to upper bound $\Delta EC_t(\bzero)$.
Here we show the computation of $\Delta EC_t(\bzero)$ for $t=1$ and using the recursion argument in the proof of Lemma~\ref{lem:delta-ec-exp}, we can show that
\begin{align}\label{eq:multi-dom-recur}
    \Delta EC_t(\beta) = \frac{C}{n}\frac{1-(1-p)^n}{ p}\sum_{i=1}^t (1-p)^i-\alpha prD\sum_{i=0}^{t-1}(1-p)^i.
\end{align}
After that, we use the same argument in the proof of Theorem~\ref{thm:multi-ber-dom} to obtain the threshold for the first day and Theorem~\ref{thm:multi-gen-dom} follows.
Now we prove the statement for $t=1$.
If the net consumption for the last day exceeds $\alpha D$ (which happens with probability $p$), then the difference between the penalty payments is $(\alpha-\beta)rD$.
Otherwise the player can save some regular payment by reporting some $\beta' D$ where $\beta \le \beta' \le \alpha$.
Let $X$ denote the number of players being busted beside the target player. 
Then $X\sim Bin(n-1,p)$ and $P(X=k)=\binom{n-1}{k}p^k(1-p)^{n-1-k}$.
Therefore,
\begin{align*}
    \Delta EC_1(\beta)&= EC_t(\alpha D) - EC_t(\beta D) \\
    &\le \max_{\beta\le \beta' \le \alpha} (1-p)\sum_{k=0}^{n-1}P(X=k)\left(\frac{C\alpha}{k\alpha+\alpha}-\frac{C\beta'}{k\alpha+\beta'}\right)-prD(\alpha-\beta)\\
    &\le \frac{\alpha-\beta}{\alpha}(1-p)M-(\alpha-\beta)prD,
    \end{align*}
and
\begin{align*}
   \max_{0\le \beta\le \alpha} \Delta EC_1(\beta) = (1-p)M-\alpha prD,
\end{align*}
given $r>\frac{(1-p)M}{\alpha pD}$, which is satisfied because actual threshold for $r$ in (\ref{eq:multi-gen-dom}) is higher. 
Using the recursion argument in Lemma~\ref{lem:delta-ec-exp}, we can obtain the expression (\ref{eq:multi-dom-recur}).
\endproof
%
%
%
%
%
%
%

\subsection{Proof for Theorem~\ref{thm:multi-gen-ne}}

\proof{}
Similar to the proof of Theorem~\ref{thm:multi-exist}, we only need to show that players who had a $\alpha$-truthful history would stay truthful.
We redefine $\Delta EC_t(\bm{b}_{t+1}^{-i})$ as in the proof of Theorem~\ref{thm:multi-gen-dom} and use a similar argument in Theorem~\ref{thm:multi-exist} to show that $\Delta EC_t(\beta D,\tilde{\bm{\alpha}}\bD)\le 0$ for $0\le \beta \le \alpha$ and $\tilde{\bm{\alpha}}\ge\bm{\alpha}$.
Then we can safely assume that players $j \ne i$ stays $\alpha$-truthful in the entire game.
Now we compare player $i$'s expected savings and penalties by reporting some $\beta D$ from being $\alpha$-truthful.
\begin{align*}
        \bbE[\text{savings}] &\le \left\{\frac{C\cdot \alpha D}{(n-1)\alpha D+\alpha D}-\frac{C\cdot \beta D}{(n-1)\alpha D+\beta D}\right\}\cdot\bbE[\text{\# days before busted}]\\
        \bbE[\text{lying}] &\ge rD(\alpha-\beta)\cdot\text{Pr(busted)}.
\end{align*}
Expected penalties exceed expected savings when  $r=\frac{1}{\alpha}\frac{C}{nD}\frac{1-(1-p)^{t+1}}{p-p(1-p)^t}$. 
\endproof

\section{Recursion Approach} \label{app:recur}
In Section~\ref{sec:single-ber}, we briefly mentioned that we can solve for the optimal cost for the Bernoulli distribution via recursion.
In the recursion proof, we compute explicitly the expression for $OptCost(t,r,b_{t+1})$ for $t<T$ and $ExpCost(T,r,b_T)$ for the first round.
Here, we provide such expressions and optimal strategies can be easily derived from these expressions. 
We note that we presented the alternative proof in the main article because it showcases the essence of our proposed mechanism. 
Moreover, the recursion approach would be computationally heavy for continuous distributions whereas the proof in the main body can be extended to any general distributions. 

The following is the complete proof via backward induction.
For simplicity, we set $D=1$, which does not affect the results.
We break the proof into four cases and together, the four cases paint the picture of the optimal strategy under the Bernoulli distributions for the single player model.

\begin{center}
\def\arraystretch{1.5}
    \begin{tabular}{|c|l|}
    \hline
        Case 1 & $p\le\tfrac{1}{2}$ and $r\le 1$ OR $p>\tfrac{1}{2}$ and $r\le \tfrac{1}{2p}$ \\\hline
        Case 2 & $p>\tfrac{1}{2}$ and $\frac{1}{2p} < r\le 1$\\\hline
        Case 3 & $p>\tfrac{1}{2}$ and $r>1$\\\hline
        Case 4 & $p\le \tfrac{1}{2}$ and $r>1$\\\hline
    \end{tabular}
\end{center}
\medskip


\paragraph{Case 1. $p\le\tfrac{1}{2}$ and $r\le 1$ OR $p>\tfrac{1}{2}$ and $r\le \tfrac{1}{2p}$}

\begin{lemma}
\label{lem:gen-p-r<1}
For any $1\le t<T$, when $p\le\tfrac{1}{2}$ and $r\le 1$ OR when $p>\tfrac{1}{2}$ and $r\le \tfrac{1}{2p}$, given yesterday's arbitrary report $b_{t+1}$, 
\[
OptCost(t, r,b_{t+1}) = (1-2p)rb_{t+1} + (t-1)p(1-2p)r + tp(1+r),
\]
which is achieved by setting $b_t = 0$.
If $r < 1$, $b_t = 0$ is the unique optimal report; if $r = 1$, the optimal report is any value $b_t \le b_{t+1}$.
\end{lemma}

\noindent\proof{}
We prove the lemma by induction. When $t=1$, 
\begin{align}\label{eq:gb:c1:r1}
    ExpCost(1,r,b_2,b_1) = p(1+r)+(1-p)b_1-prb_2+(1-p)r|b_2-b_1|.
\end{align}
The coefficient for $b_1$ is either $(1-p)(1-r)$ (if $b_2\ge b_1$) or $(1-p)(1+r)$ (if $b_2<b_1$). Both are non-negative for $r\le 1$. Therefore, by setting $b_1=0$, we achieved the optimal cost:

$$OptCost(1,r,b_2)=\min_{b_1} ExpCost(1,r,b_2,b_1)=r(1-2p)p_2+p(1+r).$$

\noindent Assume the lemma is true for round $t - 1 \ge 1$. 
For round $t$ and given yesterday's report $b_{t+1}$,
\begin{align*}
    &ExpCost(t,r,b_{t+1},b_t) \\
    &= p(1+r)+(1-p)b_t-prb_{t+1}+(1-p)r|b_{t+1}-b_t|\\
    &\quad +pOptCost(t-1,r,1)+(1-p)OptCost(t-1,r,b_t)\\
    &= p(1+r)+(1-p)b_t-prb_{t+1}+(1-p)r|b_{t+1}-b_t|\\
    &\quad+p[(1-2p)r+(t-2)p(1-2p)r+(t-1)p(1+r)]\\
    &\quad +(1-p)[(1-2p)rb_t+(t-2)p(1-2p)r+(t-1)p(1+r)]\\
    &= tp(1+r)+(t-1)p(1-2p)r+(1-p)[1+(1-2p)r]b_t+(1-p)r|b_{t+1}-b_t|-prb_{t+1}.
\end{align*}
The coefficient for $b_t$ is as follows 
\begin{align*}
    \begin{cases}
    (1-p)(1+(1-2p)r+r) = (1-p)(1+2(1-p)r) & b_t>b_{t+1}\\
    (1-p)(1+(1-2p)r-r) = (1-p)(1-2pr) & b_t\le b_{t+1}
    \end{cases}
\end{align*}
When $r\le \frac{1}{2p}$, both coefficients are non-negative. Therefore, choosing $b_t=0$ is optimal and the optimal cost is
\begin{align*}
    OptCost(t,r,b_{t+1})&=\min_{pt}ExpCpst(t,r,b_{t+1},b_t)\\
    &=tp(1+r)+(t-1)p(1-2p)r+(1-p)rb_{t+1}-prb_{t+1}\\&=tp(1+r)+(t-1)p(1-2p)r+(1-2p)rb_{t+1}.
\end{align*}
By induction, we proved the lemma.
\endproof

\begin{theorem}
\label{thm:gen-p<1}
If $p \le \frac{1}{2}$ and $r \le 1$, or if $p > \frac{1}{2}$ and $r \le \frac{1}{2p}$, the player's optimal strategy is lying-till-end.
\end{theorem}

\noindent \proof{}
Lemma \ref{lem:gen-p-r<1} showed that the theorem is true for every day except the first day. We now show that the theorem is true for the first day. 
\begin{align*}
    ExpCost(T,r,b_T) &= p(1+OptCost(T-1,r,1))+(1-p)(b_T+OptCost(T-1,r,b_T))\\
    &= p[1+(1-2p)r+(T-2)p(1-2p)r+(T-1)p(1+r)]\\
    &\quad +(1-p)[b_T+(1-2p)rb_T+(T-2)p(1-2p)r+(T-1)p(1+r)]\\
    &= (T-1)p(1+r)+(T-2)p(1-2p)r+p+(1-p)(1+r-2pr)b_T+(1-2p)pr
\end{align*}
The coefficient for $b_T$ is non-negative in both cases. So the optimal choice for the first day is also zero. Along with the Lemma \ref{lem:gen-p-r<1}, we've shown the optimal strategy is lying-till-end for $p\le\tfrac{1}{2}, r\le 1$ and $p>\tfrac{1}{2},r\le\tfrac{1}{2p}$ with optimal cost 
\begin{align*}
    OptCost(T,r)&=\min_{b_T}ExpCost(T,r,b_T)\\
    &=(T-1)p(1+r)+(T-2)p(1-2p)r+p+(1-2p)pr. 
\end{align*}
\endproof


\paragraph{Case 2. $p>\tfrac{1}{2}$ and $\frac{1}{2p} < r\le 1$}

\

\noindent When $p>\tfrac{1}{2}$ and $\frac{1}{2p} < r \le 1$, as we have seen in Equation (\ref{eq:gb:c1:r1}), 
\begin{align*}
    OptCost(1, r, b_2) = (1-2p)rb_2 + p(1+r),
\end{align*}
by setting $b_1 = 0$. Next we consider round $2 \le t<T$.

\begin{lemma}
\label{lem:gen-p-r<1-case2}
For any $2 \le t<T$, when $p>\tfrac{1}{2}$ and $\frac{1}{2p} < r \le 1$, given yesterday's arbitrary report $b_{t+1}$, 
\begin{align*}
    OptCost(t, r, b_{t+1}) &= \left[2(1-p)^tr + \sum_{l=2}^{t-1}(1-p)^l + (1 - p - r) \right]b_{t+1}+const.,
\end{align*}
which is achieved by setting $b_t = b_{t+1}$.
\end{lemma}

\noindent \proof{} 
We prove the lemma by induction. When $t=2$, 
\begin{align*}
    &ExpCost(2,r,b_3,b_2) \\
    &= \begin{cases}
    (1-p)[1 + (1-2p)r - r]b_2 + (1-2p)rb_3 + 2p(1+r) + p(1-2p)r, & b_2 \le b_3\\
    (1-p)[1 + (1-2p)r + r]p_2 - rb_3 + 2p(1+r) + p(1-2p)r, & b_2 > b_3
    \end{cases}
\end{align*}
Since $r > \frac{1}{2p}$, $1 + (1-2p)r - r \le 0$ and $1+(1-2p)r+r\ge 0$.
Thus $ExpCost(2,r,b_3,b_2)$ is a valley function with respect to $b_2$ and takes minimum by setting $b_2=b_3$.
Therefore, $OptCost$ can be written as
\begin{align*}
    OptCost(2,r,b_3) = 
    [2(1-p)^2r + (1-p-r)]b_3 + 2p(1+r) + p(1-2p)r.
\end{align*}

\noindent Assume up to round $t - 1 \ge 1$, the lemma holds. For round $t$ and yesterday's report $b_{t+1}$,
\begin{align*}
    &ExpCost(t,r,b_{t+1},b_t)\\
    &= 
    p(1+r) + (1-p)b_t - prb_{t+1}+(1-p)r|b_{t+1} - b_t|\\
    &~~~~~+ (1-p)\left[2(1-p)^{t-1}r + \sum_{l=2}^{t-2}(1-p)^l + (1 - p - r) \right]b_{t}+ p(1-2p)r\\
    &~~~~~ + (t-1)p(1+r) + (t-3)p(1-p-r) + \sum_{i = 2}^{t-2}\left[2(1-p)^ir + \sum_{l=2}^{i-1}(1-p)^l\right]\\
    &~~~~~+ p\left[2(1-p)^{t-1}r + \sum_{l=2}^{t-2}(1-p)^l + (1 - p - r)\right]\triangleq M(b_{t+1},b_t)+const.,
\end{align*}
where 
\begin{align*}
    M(b_{t+1},b_t) &= (1-p)b_t - prb_{t+1}+(1-p)r|b_{t+1} - b_t|\\
    &\quad + (1-p)\left[2(1-p)^{t-1}r + \sum_{l=2}^{t-2}(1-p)^l + (1 - p - r) \right]b_{t}.
\end{align*}
If $b_t \le b_{t+1}$, 
\begin{align*}
    M(b_{t+1},b_t)&=  (1-p)b_t - prb_{t+1}+(1-p)r|b_{t+1} - b_t|\\
    &\quad + (1-p)\left[2(1-p)^{t-1}r + \sum_{l=2}^{t-2}(1-p)^l + (1 - p - r) \right]b_{t}\\
    &= (1-p)\left[1 - r + 2(1-p)^{t-1}r + \sum_{l=2}^{t-2}(1-p)^l + (1 - p - r) \right]b_{t} + (1-2p)rb_{t+1}.
\end{align*}
Note that the coefficient of $b_t$ is $(1-p)$ times the following
\begin{align*}
    &1 - r + 2(1-p)^{t-1}r + \sum_{l=2}^{t-2}(1-p)^l + (1 - p - r)\\
    &= 1+ \sum_{l=1}^{t-2}(1-p)^l - 2r[1-(1-p)^{t-1}]= \sum_{l=0}^{t-2}(1-p)^l(1-2pr) \le 0,
\end{align*}
where the inequality is due to $r \ge \frac{1}{2p}$.

\noindent If $b_t > b_{t+1}$, 
\begin{align*}
    M(b_{t+1},b_t)
    &= (1-p)\left[1 + r + 2(1-p)^{t-1}r + \sum_{l=2}^{t-2}(1-p)^l + (1 - p - r) \right]b_{t} - rb_{t+1},
\end{align*}
where the coefficient of $b_t$ is positive.
Thus the minimum of $M(b_{t+1},b_t)$ is achieved at $b_t = b_{t+1}$, i.e.,
\begin{align*}
    OptCost(t, r, b_{t+1}) &= \left[2(1-p)^tr + \sum_{l=2}^{t-1}(1-p)^l + (1 - p - r) \right]b_{t+1}+const.
\end{align*}
By induction, we proved the lemma.
\endproof

\begin{theorem}\label{thm:case2}
When $p>\tfrac{1}{2}$ and $\frac{1}{2p} < r \le 1$, the optimal strategy is lying-till-busted for the first $T-1$ rounds and lying in the last round.
\end{theorem}

\noindent\proof{}
Let us consider the first day. 
\begin{align*}
    ExpCost(T,r,b_T) &= p(1+OptCost(T-1,r,1))+(1-p)(b_T+OptCost(T-1,r,p_T))\\
    &= (1-p)\left[1+2(1-p)^tr + \sum_{l=2}^{t-1}(1-p)^l + (1 - p - r)\right]b_T + const.
\end{align*}
The coefficient for $b_T$ is positive when $\frac{1}{2p} < r < 1$, thus $b_T = 0$.
\endproof


\paragraph{Case 3. $p > \frac{1}{2}$ and $r > 1$}

\begin{lemma}
\label{lem:gen-p-r>1}
For $p > \frac{1}{2}$, $r > 1$, and any $1\le t<T$, given yesterday's arbitrary report $b_{t+1}$, 
\begin{align*}
    OptCost(t,r,b_{t+1}) = \left[(1-p-pr)\sum_{i=0}^{t-1}(1-p)^i\right]b_{t+1}+const.,
\end{align*}
which is achieved by setting $b_t=b_{t+1}$. 
\end{lemma}
 
\noindent\proof{}
We prove the lemma by induction. When $t=1$, the expected cost is 
\begin{align*}
    ExpCost(1,r,b_2,b_1)=p(1+r)+(1-p)b_1-prb_2+(1-p)r|b_b-b_1|.
\end{align*}
The coefficient for $b_1$ is $(1-p)(1+r)$ for $b_1\ge b_2$ and is positive. 
The coefficient is $(1-p)(1-r)$ for $b_1<b_2$ and is negative. 
This implies that $ExpCost(1,t,b_2,b_1)$ is a valley function and the minimum is achieved by setting $b_1=b_2$. Thus the optimal cost for $t=1$ is 
\begin{align*}
    OptCost(1,r,b_2)=p(1+r)+(1-p-pr)b_2.
\end{align*}
Assume up to round $t-1\ge1$, the lemma holds. 
For round $t$ and yesterday's report $b_{t+1}$, 
\begin{align*}
    &ExpCost(t,r,b_{t+1},b_t)\\
    &= p(1+r)+(1-p)b_t - prb_{t+1}+(1-p)r|b_{t+1} - b_t| + p(t-1)\\
    &\quad +(1-p)\left[b_{t}(1-p-pr)\sum_{i=0}^{t-2}(1-p)^i+(1+r)p\sum_{i=0}^{t-2}(1-p)^i+t-1-\sum_{i=0}^{t-2}(1-p)^i\right]\\
    &= M(b_{t+1},b_t)+const.,
\end{align*}
where 
\begin{align*}
    M(b_{t+1},b_t) &= (1-p)b_t-prb_{t+1}+(1-p)r|b_{t+1}-b_t|\\
    &\qquad+p_t(1-p-pr)(1-p)\sum_{i=0}^{t-2}(1-p)^i.
\end{align*}
When $b_t\ge b_{t+1}$, the coefficient for $b_t$ is as follows
\begin{align*}
    (1-p)\left\{1+(1-p)\sum_{i=0}^{t-2}(1-p)^i-pr\sum_{i=0}^{t-2}(1-p)^i+r\right\}&=  (1-p)\left\{\sum_{i=0}^{t-1}(1-p)^i+r\left[1-p\sum_{i=0}^{t-2}(1-p)^i\right]\right\}\\
    &= (1-p)\left\{\sum_{i=0}^{t-1}(1-p)^i+r(1-p)^{t-1}\right\},
\end{align*}
which is always positive. 
When $b_t<b_{t+1}$, the coefficient is as follows
\begin{align*}
   & (1-p)\left\{1+(1-p)\sum_{i=0}^{t-2}(1-p)^i-pr\sum_{i=0}^{t-2}(1-p)^i+r\right\}\\
   &= (1-p)\left\{1+(1-p)\sum_{i=0}^{t-2}(1-p)^i-r\left[1+p\sum_{i=0}^{t-2}(1-p)^i\right]\right\},
\end{align*}
which is negative when 
\begin{align}
    r &> \frac{1+(1-p)\sum_{i=0}^{t-2}(1-p)^i}{1+p\sum_{i=0}^{t-2}(1-p)^i} = \frac{\sum_{i=0}^{t-1}(1-p)^i}{2-(1-p)^{t-1}}=\frac{1-(1-p)^t}{2p-p(1-p)^{t-1}} \label{eq:gen-p-r>1}
\end{align}
Note that from Equation (\ref{eq:gen-p-r>1}), we see when $p> \tfrac{1}{2}$, the right-hand-side is smaller than $1$. Thus given $r > 1$, the $M$ function is a valley function, and the minimum is achieved by setting $b_t=b_{t+1}$. The optimal cost in round $t$ is then
\begin{align*}
    OptCost(t,r,b_{t+1})
    &= M(b_{t+1},b_{t+1})+const.\\
    &= (1-p-pr)\left[1+(1-p)\sum_{i=0}^{t-2}(1-p)^i\right]b_t+const.\\
    &= \left[(1-p-pr)\sum_{i=0}^{t-1}(1-p)^i\right]b_t+const.
\end{align*}
By induction, we proved the lemma. 
\endproof 

\begin{theorem}\label{thm:case3}
When $p > \frac{1}{2}$, if $r\ge \frac{1-(1-p)^{T}}{p(1-(1-p)^{T-1})}$, honest-till-end is the optimal strategy; if $1<r<\frac{1-(1-p)^{T}}{p(1-(1-p)^{T-1})}$, lying-till-busted is optimal.
\end{theorem}

\noindent\proof{}
We write out the expected cost on the first round, i.e.,  $t=T$.
\begin{align*}
    ExpCost(T,r,b_T) &= p(1+OptCost(T-1,r,1))+(1-p)(b_T+OptCost(T-1,r,b_T))\\
    &= (1-p)\left[1+(1-p-pr)\sum_{i=0}^{T-2}(1-p)^i\right]b_T + const.\\
    &= (1-p)\left[1+(1-p)\sum_{i=0}^{T-2}(1-p)^i-pr\sum_{i=0}^{T-2}(1-p)^i\right]b_T + const.\\
    &= (1-p)\left[\sum_{i=0}^{T-1}(1-p)^i-pr\sum_{i=0}^{T-2}(1-p)^i\right]b_T + const.
\end{align*}
The coefficient for $b_T$ is positive when 
\begin{align}\label{eq:honest-r}
    r<\frac{\sum_{i=0}^{T-1}(1-p)^i}{p\sum_{i=0}^{T-2}(1-p)^i}=\frac{\tfrac{1-(1-p)^T}{p}}{p\tfrac{1-(1-p)^{T-1}}{p}}=\frac{1-(1-p)^{T}}{p(1-(1-p)^{T-1})}.
\end{align}
The optimal strategy for the first day is therefore setting $b_T=0$ when $r$ smaller than (\ref{eq:honest-r}) and $b_T=1$ otherwise.
Along with Lemma~\ref{lem:gen-p-r>1}, we have proved the theorem.
\endproof

\paragraph{Case 4. $p \le \frac{1}{2}$ and $r \ge 1$}

\
\medskip

\noindent For any $2 \le t \le T-1$, let
\begin{align*}
    h(t) = \frac{\sum_{i=0}^{t-1}(1-p)^i}{1+p\sum_{i=0}^{t-2}(1-p)^i} = \frac{1-(1-p)^t}{2p-p(1-p)^{t-1}},
\end{align*}
and 
\begin{align*}
    A(t) = (1-r)\sum_{i=1}^{t-1}(1-p)^i + (1-p-pr)(1-p)^{t-1} + (1-2p)r\sum_{i=1}^{t-1} (1-p)^{i-1}.
\end{align*}

\begin{claim}\label{clm:mnt}
When $p < \frac{1}{2}$, $1=h(1) < h(2) < \cdots < h(T-1) < h(T)<\tfrac{1}{2p}$.
\end{claim}

\noindent\proof{}
The derivative of $h(t)$ with respect to $t$ is strictly positive:
\begin{align*}
    \frac{d}{dt}h(t) 
    &= \frac{d}{dt}\frac{1-(1-p)^t}{2p-p(1-p)^{t-1}} \\
    &= \frac{[2p-p(1-p)^{t-1}][-(1-p)^t\ln(1-p)]-[1-(1-p)^t][-p(1-p)^{t-1}\ln(1-p)]}{[2p-p(1-p)^{t-1}]^2}\\
    &= \frac{p(1-p)^{t-1}\ln(1-p)[1-2p(1-p)]}{[2p-p(1-p)^{t-1}]^2}\\[4pt]
    &> \frac{p(1-p)^{t-1}\ln(1-p)[1-2p]}{[2p-p(1-p)^{t-1}]^2}\ge 0
\end{align*}
The edge cases can be checked manually. Thus, $h(t)$ is increasing w.r.t. $t$. 
\endproof

\begin{claim}\label{clm:Ak}
When $p < \frac{1}{2}$ and $t \ge 2$,
$A(t) + (1-r) \ge 0$ if and only if
and $r \le h(t+1)$.
\end{claim}

\noindent \proof{}
We simplify the expression $A(t)+1-r$ as follows.
\begin{align*}
    A(t)+1-r &= (1-r)\sum_{i=0}^{t-1}(1-p)^i+(1-p-pr)(1-p)^{t-1}+(1-2p)r\sum_{i=1}^{t-1}(1-p)^{i-1}\\
    &= \frac{1-(1-p)^t}{p}+(1-p)^t-r(2-(1-p)^t).
\end{align*}
Thus $A(t)+1-r\ge 0$ if and only if 
\begin{align*}
    r\le \frac{\tfrac{1-(1-p)^t}{p}+(1-p)^t}{2-(1-p)^t} =\frac{1-(1-p)^{t+1}}{2p-p(1-p)^t}=h(t+1). 
\end{align*}
\endproof

\medskip

\noindent Next, we further distinguish the following subcases $h(t-1)< r \le h(t)$ for each $t = 2, \dots, T-1$.

\paragraph{SubCase 4.1. $p \le \frac{1}{2}$ and $h(t-1)< r \le h(t)$}

\

\medskip
\noindent We start with the last day,
\begin{align*}
ExpCost(1, r, b_2, b_1) = 
\begin{cases}
(1-p)(1-r)b_1 + (1-2p)rb_2 + const., & b_2 \ge b_1\\
(1-p)(1+r)b_1 -rb_2 + const., & b_2 < b_1.
\end{cases}
\end{align*}
Given $r \ge h(t-1) \ge 1$,  
$OptCost(1, r, b_2) = (1-p-pr)b_2 + const.$,
achieved by setting $b_1=b_2$.
\medskip

\noindent Now we consider the rounds after the first $t$ rounds, i.e., the last $T-t$ rounds. 

\begin{lemma}
\label{lem:b:sp:sk}
When $p < \frac{1}{2}$ and $h(t-1)< r \le h(t)$, given yesterday's arbitrary report $b_{t'+1}$, 
\begin{align}
    OptCost(t',r,b_{t'+1}) &= A(t') b_{t'+1} + const.,
\end{align}
which is achieved by setting $b_{t'}=b_{t'+1}$ for any $2\le t' < t$ and by setting $b_{t'}=0$ for $t\le t'\le T-1$. 
\end{lemma}

\noindent \proof{}
For $t'=2$,
\begin{align*}
ExpCost(2, r, b_3, b_2)
= 
\begin{cases}
(1-p)[1-r + (1-p-pr)]b_2 + (1-2p)rb_3 + const, & b_3 \ge b_2\\
(1-p)[1+r + (1-p-pr)]b_2 -rb_3 + const, & b_3 < b_2.
\end{cases}
\end{align*}
Given $p < \frac{1}{2}$ and $r >h(t-1)\ge h(2)=  \frac{2-p}{1+p}$, $ExpCost(2, r, b_3, b_2)$ is a valley function and takes minimum at $b_2=b_3$.
Thus, 
\begin{align*}
OptCost(2, r, b_3) = (2-p)(1-p-pr)b_3 + const. = A(2)b_3+const.
\end{align*}
In general, for any $ t' \ge 3$, 
\begin{align*}
ExpCost(t', r, b_{t'+1}, b_{t'}) = 
\begin{cases}
(1-p)[1-r + A(t'-1)]b_{t'} + (1-2p)rb_{t'+1} + const., & b_{t'+1} \ge b_{t'}\\
(1-p)[1+r + A(t'-1)]b_{t'} -rb_{t'+1} + const., & b_{t'+1} < b_{t'}.
\end{cases}
\end{align*}

By Claim~\ref{clm:Ak}, for $2\le t'< t$, $1-r+A(t'-1)\le 0$ since $r> h(t-1)> h(t'-1)$.
Then $ExpCost(t,r,b_{t'+1},b_{t'})$ is a valley function and takes minimum at $b_{t'}=b_{t'+1}$.
For $t\le t'\le T-1$, $1-r+A(t'-1)\ge 0$ since $r\le h(t)\le h(t')$.
Then the coefficient for $b_{t'}$ in both cases is positive and the expected cost takes minimum at $b_{t'}=0$.
\endproof

\noindent Finally, we consider the first day,
\begin{align*}
ExpCost(T, r, p_T) = (1-p)[1 + A(T-1)]b_T + const.,
\end{align*}
where the coefficient for $b_T$ is positive. 
Thus on the first day, the optimal $b_T=0$.
In conclusion, we have the following theorem.
\begin{theorem}
When $p < \frac{1}{2}$, $2 \le t \le T-1$, and $h(t-1)< r \le h(t)$,
the optimal strategy is lying-till-end for the first $t$ rounds, and lying-till-busted for the rest of the game.
\end{theorem}

\paragraph{SubCase 4.2. $r \ge h(T-1)$}

\

\medskip
\noindent When $r \ge h(T-1)$, as we have seen in previous subcase, 
\begin{align*}
    ExpCost(T,r,b_T) = (1-p)[1+A(T-1)]b_T + const.,
\end{align*}
where
\begin{align*}
        1+A(T-1) 
        &= 1+ (1-r)\sum_{i=1}^{T-2}(1-p)^i + (1-p-pr)(1-p)^{T-2}+ (1-2p)r\sum_{i=1}^{T-2}(1-p)^{i-1}\\
        &= 1 + (1-p-pr)\sum_{i=1}^{T-1} (1-p)^{i-1}.
\end{align*}
Thus $1+A(T-1) \ge 0$ if and only if
\[
r \le \frac{1-(1-p)^T}{p(1-(1-p)^{T-1})}.
\]
In conclusion, we have the following theorem. 

\begin{theorem}
When $p < \frac{1}{2}$, if $\frac{1-(1-p)^{T-1}}{2p-p(1-p)^{T-2}} < r \le \frac{1-(1-p)^T}{p(1-(1-p)^{T-1})}$, the optimal strategy is lying-till-busted; 
if $r > \frac{1-(1-p)^T}{p(1-(1-p)^{T-1})}$, the optimal strategy is honest-till-end.
\end{theorem}

\end{document}